\documentclass[10pt,conference]{IEEEtran}

\usepackage[utf8]{inputenc}
\usepackage[colorlinks,bookmarksopen,bookmarksnumbered,citecolor=black,urlcolor=black,linkcolor=black]{hyperref}
\usepackage{mathtools}
\usepackage{amsthm}
\usepackage{amssymb}
\usepackage{amsmath}
\usepackage{yhmath}
\usepackage{amsfonts}
\usepackage{url}
\usepackage{listings}
\usepackage{color}
\usepackage[dvipsnames]{xcolor}
\usepackage{graphicx}
\usepackage{caption}
\usepackage[noadjust]{cite}
\usepackage{breqn}
\usepackage{makecell}
\usepackage{enumitem}
\usepackage[titlenumbered,ruled,linesnumbered]{algorithm2e}
\usepackage{balance}
\usepackage{amssymb}
\usepackage{multirow}
\usepackage{multicol} 
\usepackage{soul}
\usepackage{subfigure}
\usepackage{adjustbox}
\usepackage{arydshln}

\newtheoremstyle{claim}
  {\topsep}
  {\topsep}
  {}
  {}
  {\itshape}
  {.}
  {.5em}
  {\thmname{#1}\thmnumber{ #2}\thmnote{ (#3)}}
\theoremstyle{claim}

\newtheorem{constraint}{Constraint}
\newtheorem{definition}{Definition}
\newtheorem{theorem}{Theorem}
\newtheorem{lemma}{Lemma}

\newtheorem{example}{Example}

\definecolor{cmdcolor}{RGB}{34, 153, 84}

\newcommand{\eg}{{\it e.g.}\xspace}
\newcommand{\ie}{{\it i.e.}\xspace}

\usepackage{tikz}

\newcommand{\blackcircled}[1]{%
\tikz[baseline=(char.base)]{
\node[shape=circle,draw,inner sep=0.5pt,fill=black, text=white] (char) {#1};}}

\IEEEoverridecommandlockouts

\begin{document}

\thispagestyle{plain}
\pagestyle{plain}

\newcommand\pname{FRAP}






\title{\pname{}: A Flexible Resource Accessing Protocol for Multiprocessor Real-Time Systems\thanks{This work is supported by the National Natural Science Foundation of China (NSFC) under Grant 62302533. Corresponding author: Wanli Chang,  wanli.chang.rts@gmail.com.}}


\author{
\IEEEauthorblockN{Shuai Zhao\IEEEauthorrefmark{2}, Hanzhi Xu\IEEEauthorrefmark{2}, Nan Chen\IEEEauthorrefmark{3}, Ruoxian Su\IEEEauthorrefmark{2}, Wanli Chang\IEEEauthorrefmark{4}}
\IEEEauthorblockA{
\IEEEauthorrefmark{2}Sun Yat-sen University, China
\IEEEauthorrefmark{3}University of York, UK
\IEEEauthorrefmark{4}Hunan University, China
}
}

\maketitle
\begin{abstract}
Fully-partitioned fixed-priority scheduling (FP-FPS) multiprocessor systems are widely found in real-time applications, where spin-based protocols are often deployed to manage the mutually exclusive access of shared resources.
Unfortunately, existing approaches either enforce rigid spin priority rules for resource accessing or carry significant pessimism in the schedulability analysis, imposing substantial blocking time regardless of task execution urgency or resource over-provisioning.
This paper proposes \pname{}, a spin-based flexible resource accessing protocol for FP-FPS systems.
A task under \pname{} can spin at any priority within a range for accessing a resource, allowing flexible and fine-grained resource control with predictable worst-case behaviour. 
Under flexible spinning, we demonstrate that the existing analysis techniques can lead to incorrect timing bounds and present a novel MCMF (minimum cost maximum flow)-based blocking analysis, providing predictability guarantee for \pname{}.
A spin priority assignment is reported that fully exploits flexible spinning to reduce the blocking time of tasks with high urgency, enhancing the performance of \pname{}. 
Experimental results show that \pname{} outperforms the existing spin-based protocols in schedulability by 15.20\%-32.73\% on average, up to 65.85\%. 
\end{abstract}

\IEEEpeerreviewmaketitle

\section{Introduction} 
\label{sec:intro}

The increasing demand for emerging real-time applications has necessitated the transition from single processor to multiprocessor systems, with a fully-partitioned fixed-priority scheduling (FP-FPS) scheme commonly applied in practice~\cite{davis2011survey,chang2016resource,chang2019java,chang2020development}. 
On such systems, parallel tasks often need to operate on shared objects, \eg, memory blocks, code segments and I/O ports.
However, these shared resources must be accessed mutually exclusively to ensure data integrity, which causes unbounded blocking that compromises system timing predictability.
To address this, multiprocessor resource sharing protocols~\cite{brandenburg2022multiprocessor} are developed that regulate access to shared resources, providing the worst-case blocking guarantee.

Existing multiprocessor resource sharing protocols can be classified into two categories based on the locking primitives: suspension-based and spin-based~\cite{burns2001real}. 
With the suspension-based approach~\cite{rajkumar1988real,brandenburg2013omlp}, a task is switched away by the scheduler if its resource request is not immediately satisfied.
By contrast, tasks under spin locks actively wait (spin) for a resource on their processors, until the resource is granted~\cite{burns2013schedulability}.
As discussed in~\cite{block2007flexible}, the suspension-based approaches can cause frequency context switches with a non-trivial overhead, while the spin-based ones can impose a delay to other tasks on the processor.
However, the spin locks have the advantage of low complexity and are largely applied in industrial systems~\cite{zhao2020priority,burns2013schedulability}, \eg, it is mandated by the AUTOSAR standard~\cite{furst2009autosar}.

Based on spin locks, various resource sharing protocols are developed, \eg, 
the Multiprocessor Stack Resource Protocol (MSRP)~\cite{gai2001minimizing,gai2002stack} and the Preemptable Waiting Locking Protocol (PWLP)~\cite{alfranseder2014efficient}.
Each protocol defines a set of uniform rules that regulate the resource accessing behaviours of all tasks, including the priorities at which tasks should spin and execute with resources. For instance, tasks under MSRP spin non-preemptively for resources, whereas PWLP defines that tasks spin at their base priorities with preemption enabled.

However, in real-world systems, tasks (with different execution urgency) often have various demands on different resources, \eg, resources with diverse critical section lengths.
Unfortunately, existing solutions enforce a uniform rule for the spin priorities of all tasks, imposing a spin delay on tasks regardless of their urgency, leading to deadline misses that jeopardise system schedulability.
As shown in~\cite{zhao2017new}, MSRP is not favourable for long resources as the resource accessing task can impose a long blocking on others.
By contrast, a task under PWLP can incur a significant delay if it is frequently preempted when spinning for a resource.
Hence, when facing various resource accessing scenarios in a system, the effectiveness of existing spin-based protocols can be significantly compromised as spin priorities are assigned by a single rule that considers neither the requesting task nor the resource.

Alternatively,~\cite{afshar2014flexible,afshar2017optimal} presents a resource sharing solution that employs both spin and suspension approaches to provide flexible resource control. However, the suspension approach causes multiple priority inversions to tasks, imposing severe blocking that endangers their timing requirements~\cite{brandenburg2022multiprocessor}. More importantly, the hybrid of both types of locks significantly complicated the corresponding schedulability analysis, leading to overly pessimistic analytical bounds that undermine the system performance~\cite{block2007flexible,afshar2017optimal}. 

\textbf{Contributions:} This paper presents a novel multiprocessor Flexible Resource Accessing Protocol (\pname{}) for FP-FPS systems, which serves resources in a first-in-first-out (FIFO) order managed solely by spin locks. 
In addition, unlike existing spin-based protocols, \pname{} allows tasks to spin at any priority between the base priority and the highest priority (\ie, effectively non-preemptive) for accessing a resource. This enables flexible and fine-grained resource sharing control for different resource accessing scenarios. 
With flexible spin priority enabled, we demonstrate an analysis problem of \pname{} which is not seen in existing protocols and 
construct a new MCMF (minimum cost maximum flow)-based blocking analysis~\cite{ahuja1995network}, providing timing predictability for \pname{}.
Then, a spin priority assignment is proposed that computes the spin priorities of each task when accessing every resource it requests, which enhances the performance of \pname{} by reducing the blocking of tasks with high urgency. 
The experimental results show \pname{} outperforms existing spin-based protocols by $15.20\%-32.73\%$ (up to $65.85\%$) on average in system schedulability, and significantly outperforms the hybrid locking approach~\cite{afshar2014flexible,afshar2017optimal} by $76.47\%$ on average (up to $1.78$x).
With \pname{} constructed, we overcome the limitations of existing solutions and provide an effective solution for managing resources under a wide range of accessing scenarios.

\textbf{Organisation:} Sec. \ref{sec:model} describes the system model. Sec. \ref{sec:related} provides the related work and the motivation. Sec. \ref{sec:protocol} describes the working mechanism of \pname{}. Then, Sec. \ref{sec:RTA} presents the proposed blocking analysis for \pname{} and Sec.~\ref{sec:assignment} determines the spin priorities of tasks. Finally, Sec.~\ref{sec:evaluation} presents the experimental results and Sec. \ref{sec:conclusion} draws the conclusion.
 
\section{System Model} 
\label{sec:model}


We focus on a multiprocessor system that contains a set of symmetric processors $\Lambda$ and a set of sporadic tasks $\Gamma$. 
A processor is denoted by $\lambda_m \in \Lambda$, and $\tau_i \in \Gamma$ indicates a task.
Tasks are statically allocated and prioritised on each processor by mapping and priority ordering algorithms before execution, and are scheduled by the FP-FPS scheme at runtime. 
Notation $\Gamma_{m}$ denotes the set of tasks on $\lambda_m$.

A task $\tau_i$ is defined as $\tau_i=\{C_i, T_i, D_i, P_i, A_i\}$, in which $C_i$ is the pure Worst-Case Execution Time (WCET) without accessing any resource, $T_i$ is the period, $D_i$ is the constrained deadline, $P_i$ is the base priority without accessing a resource, and $A_i$ is the allocation of $\tau_i$. Each task has a unique priority, where a higher numeric value indicates a higher priority. Function $\text{lhp}(i)=\{\tau_h ~|~ A_h = A_i \wedge P_h > P_i\}$ and $\text{llp}(i)=\{\tau_l ~|~ A_l = A_i ~\wedge~ P_l < P_i\}$ denotes the set of local high-priority and low-priority tasks of $\tau_i$, respectively. Notation $\tau_j$ is a remote task of $\tau_i$, \ie, $\tau_j \in \Gamma_{m}, \lambda_m \neq A_i$.


In addition, the system contains a set of mutually-exclusively shared resources, denoted as $\mathbb{R}$. Accesses to a resource $r^k \in \mathbb{R}$ are protected by a spin lock, \ie, tasks busy-wait (spin) if $r^k$ is not available. A resource $r^k$ has a critical section length of $c^k$, indicating the worst-case computation time of tasks on $r^k$. The number of times that $\tau_i$ can request $r^k$ in one release is denoted as $N_i^k$. 
Function $F(\cdot)$ returns the resources requested by the given tasks, \eg, $F(\tau_1) = \{r^1, r^2\}$ indicates $\tau_1$ requires $r^1$ and $r^2$ during execution.
When a task requires a locked resource, it spins with the designated spin priority until the resource is granted. 
The spin priority of $\tau_i$ for accessing $r^k$ is given by $P_i^k$, specifying at which priority that $\tau_i$ should busy wait for $r^k$. 
The spin priorities are computed by Alg.~\ref{alg:assign} in Sec.~\ref{sec:assignment-slack}.
Resources shared across processors are termed \textit{global resources}, whereas ones that are shared within one processor are \textit{local resources}.
We assume a task can hold at most one resource at a time, however, nested resources can be supported using group locks~\cite{block2007flexible,zhao2017new}. 
Notations introduced in the system model are given in Tab.~\ref{tab:model_notation_table}.



\section{Existing Resource Sharing Protocols}
\label{sec:related}


Numerous resource sharing protocols with spin locks applied are available for FP-FPS systems~\cite{zhao2018thesis}. However, as shown in~\cite{wieder2013spin,zhao2017new,zhao2018thesis}, there exists no optimal solution that dominates others, in which each protocol favours systems with certain task and resource characteristics. 
This section describes the working mechanism of the mainstream resource sharing solutions (Sec.~\ref{sec:related-protocol}) and their timing analysis (Sec.~\ref{sec:related—analysis}).


\subsection{Resource Sharing Protocols with Spin Locks}\label{sec:related-protocol}

The majority of spin-based resource sharing protocols in FP-FPS systems serve resources in a FIFO order, each providing a set of rules to manage the global resources~\cite{brandenburg2022multiprocessor}. Local resources are controlled by the Priority Ceiling Protocol (PCP)~\cite{goodenough1988priority} in FP-FPS systems.
A comprehensive review can be found in~\cite{brandenburg2022multiprocessor} and~\cite{zhao2018thesis}. 
Below we detail the working mechanism of the mainstream protocols with spin locks applied.



The Multiprocessor Stack Resource Protocol (MSRP)~\cite{gai2001minimizing,gai2002stack} provides a non-preemptive resource sharing approach with FIFO spin locks. Under MSRP, each resource is associated with a FIFO queue that specifies the order that the resource is served. 
When $\tau_i$ requests a resource $r^k$ that is currently locked by another task, it joins at the end of the FIFO queue and spins non-preemptively. 
Once $\tau_i$ becomes the head of the queue, it locks $r^k$ and remains non-preemptable. After $\tau_i$ releases $r^k$, it leaves the FIFO queue and restores the priority to $P_i$.
MSRP provides strong protection to the resource accessing task, however, it imposes an impact on the execution of the local high-priority tasks (\ie, $\text{lhp}(i)$), jeopardising their timing requirements (see Sec.~\ref{sec:related—analysis}).

The Preemptable Waiting Locking Protocol (PWLP)~\cite{alfranseder2014efficient} defines that tasks spin at their base priorities, but execute with resources non-preemptively. That is, $\tau_i$ can be preempted by $\tau_h \in \text{lhp}(i)$ when spinning for $r^k$. 
This reduces the blocking of tasks in $\text{lhp}(i)$ due to resource requests of $\tau_i$. 
However, a preemption on $\tau_i$'s resource access can block the remote tasks that are waiting for the same resource.
To avoid this delay, PWLP cancels the request of $\tau_i$ when it is preempted and removes $\tau_i$ from the FIFO queue, so that other tasks can access the resource. Once $\tau_i$ is resumed, it re-requests $r^k$ at the end of the queue.
However, with the cancellation mechanism, $\tau_i$ can incur an additional blocking due to the re-requests to $r^k$ if it is preempted when spinning for the resource.


The Multiprocessor resource sharing Protocol (MrsP)~\cite{burns2013schedulability} provides an alternative that $\tau_i$ spins and executes with $r^k$ at the ceiling priority, \ie, the highest priority among local tasks that request $r^k$. If $\tau_i$ is preempted when it is the head of the queue (\ie, eligible to use $r^k$), $\tau_i$ is migrated to a processor with a task spinning for $r^k$ (if it exists), at which $\tau_i$ resumes and executes with $r^k$ using the wasted spinning cycles. By doing so, MrsP provides an opportunity for $\tau_i$ to execute with $r^k$ when being preempted, while mitigating the impact on other tasks. However, as shown in~\cite{zhao2017new}, the cost of migrations is non-trivial, resulting in an additional delay on $\tau_i$ and other spinning tasks, especially when the preemptions are frequent.

Another relevant work is developed in~\cite{afshar2014flexible,afshar2017optimal}, which applies both spin and suspension approaches to manage resources. As described in~\cite{afshar2014flexible}, the suspension is realised by spin locks with a low spinning priority (\eg, zero), where the spinning task gives up the processor as long as another task arrives. For each resource, \cite{afshar2017optimal} provides a priority configuration method that produces an accessing priority for tasks on each processor. 
This allows a certain degree of flexibility in resource management. 
However, when a task is preempted during spinning, it is switched away while remaining in the FIFO queue, which prolongs the blocking of other tasks that request the resource~\cite{alfranseder2014efficient}.
In addition, with suspension, tasks can incur multiple priority inversion blocking with non-trivial overheads~\cite{brandenburg2022multiprocessor}. 
Most importantly, such a hybrid locking solution greatly complicated its analysis as the blocking effects caused by both locking approaches must be bounded. This leads to a high degree of pessimism where a request can be accounted for multiple times~\cite{wieder2013spin}. 
Hence, as demonstrated in Sec.~\ref{sec:evaluation}, the effectiveness of the hybrid approach is significantly undermined due to the above limitations.

In addition, there exist other protocols with spin locks~\cite{takada1997novel,block2007flexible,chen2022msrp}, but require additional implementation support (\eg, SPEPP in~\cite{takada1997novel}) or enforce other forms of locking primitives in the implementation (\eg, FMLP in~\cite{block2007flexible}). These protocols are acknowledged but are not the main focus of this work.

\begin{table}[t]
\vspace{3pt}
  \centering
  \caption{Notations introduced in Sec.~\ref{sec:model} and~\ref{sec:related}.}
  \label{tab:model_notation_table}
  \begin{tabular}{p{.18\columnwidth}p{.72\columnwidth}}
    \hline
    \textbf{Notations} & \textbf{Descriptions} \\
    \hline
    $\Lambda$ & The set of symmetric processors in the system. \\
    $\lambda_m$ & A processor with an index of $m$.\\
    \hline
    $\Gamma$~/~$\Gamma_m$ & The set of tasks in the system / tasks on $\lambda_m$. \\
    $\tau_i$ & A task with an index of $i$. \\
    ${C}_{i}$ & The pure WCET of $\tau_i$. \\
    $T_i$ / $D_i$ & The period / deadline of $\tau_i$. \\
    $P_i$ / $A_i$ & The priority / allocation of $\tau_i$. \\
    $\tau_h$ / $\tau_l$ / $\tau_j$ & A local high-priority task / a local low-priority task / a remote task of $\tau_i$. \\
    $\text{lhp}(i)$ / $\text{llp}(i)$ &  The set of local high-priority / low-priority tasks of $\tau_i$.\\
    \hline
    $\mathbb{R}$ & The set of shared resources in the system. \\
    $r^k$ & A resource with an index of $k$. \\
    ${c}^{k}$ & The critical section length of $r^k$. \\
    $N_i^k$ & The number of requests for $r^k$ by $\tau_i$ in one release. \\
    $P_i^k$ & The spin priority of $\tau_i$ for accessing $r^k$.\\
    $F(\cdot)$ & The set of resources that are required by the given tasks. \\
    \hline
    $R_i$ & The worst-case response time of $\tau_i$.\\
    $I_i$ & The worst-case high-priority interference of $\tau_i$. \\
    $E_i$ / $B_i$ / $W_i$  & Spin delay / arrival blocking / additional blocking of $\tau_i$. \\
    $\zeta_{i}^k$ / $\xi_{i,m}^k$ & The number of requests for $r^k$ during $R_i$ issued by tasks in $\tau_i \cup \text{lhp}(i)$ / tasks in $\Gamma_{m}$. \\
    \hline
  \end{tabular}
  \vspace{-9pt}
\end{table}

\subsection{Timing Bounds of FIFO Spin-based Protocols} \label{sec:related—analysis}

This section describes the analysis for FP-FPS systems with MSRP, PWLP and MrsP~\cite{wieder2013spin,zhao2018thesis}, \ie, the spin-based approaches. The notations introduced by the analysis are summarised in Tab.~\ref{tab:model_notation_table}.
Note, as the analysis of the hybrid approach mainly focuses on the blocking caused by suspension, it is not described and is referred to~\cite{afshar2014flexible,afshar2017optimal}.


\textbf{Blocking effects.} As pointed out by~\cite{wieder2013spin,zhao2018thesis}, $\tau_i$ under spin locks can incur \textit{spin delay} $E_i$, \textit{arrival blocking} $B_i$, and \textit{additional blocking} $W_i$. An example system that illustrates the blocking effects is presented in Fig.~\ref{fig:blocking_examp}, which contains four tasks that request $r^1$ on two processors. The index of tasks represents their priority.
First, $\tau_i$ can incur a spin delay directly if its request to $r^k$ is not immediately satisfied, \eg, $\tau_2$ is blocked by $\tau_4$ from $t=2$ to $t=3$. In addition, it can incur an indirect spin delay from $\tau_h \in \text{lhp}(i)$, in which $\tau_h$ preempts $\tau_i$ but is blocked for requesting a resource, which in turn, delays $\tau_i$ (\eg, $\tau_2$ is indirectly delayed by $\tau_3$ at $t=4$). The spin delay occurs in all spin-based protocols~\cite{zhao2018thesis}. Second, when $\tau_i$ is released, it can incur an arrival blocking from $\tau_l \in \text{llp}(i)$, where $\tau_l$ is spinning or executing with $r^k$ at a priority equal to or higher than $P_i$ (\eg, $\tau_2$ is blocked by $\tau_1$ from $t=0$ to $t=1$). 
Third, with preemptable spin-based protocols, $\tau_i$ can incur additional blocking if it is preempted when accessing a resource. For instance, with PWLP applied, the request of $\tau_2$ is cancelled at $t=3$ as $\tau_2$ is preempted, it then re-requests $r^1$ at $t=6$ but incurs an additional blocking caused by  $\tau_4$.

\begin{figure}[t]
\centering
\hspace{-5pt}\includegraphics[width=.95\columnwidth]{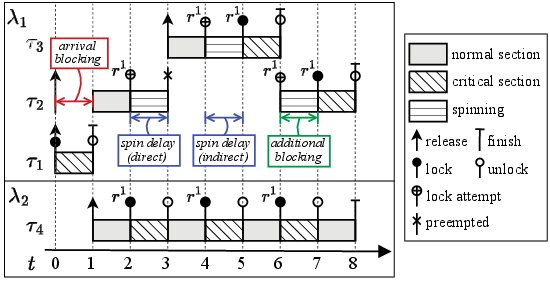}
\caption{A system for illustrating the blocking effects.}
\label{fig:blocking_examp}
\vspace{-13pt}
\end{figure}

\textbf{Timing Bound.} The response time of $\tau_i$ in FP-FPS systems with a spin-based protocol is shown in Eq.~\ref{eq:r}, in which $\overline{C_i} = C_i + \sum_{r^k \in F(\tau_i)} N_i^k \cdot c^k$ is the total execution time of $\tau_i$ (including the execution on resources), $E_i$ is the spin delay, $B_i$ is the arrival blocking, $W_i$ is the additional blocking, and $I_i = \sum_{\tau_h\in{\text{lhp}(i)}} \Big(\left\lceil{\frac{R_{i}}{T_h}}\right\rceil \cdot {\overline{C_h}}\Big)$ is the high-priority interference. 
\begin{equation}\label{eq:r}
\small
R_i  =  \overline{C_i} + E_i + B_i + W_i 
 + I_i 
\end{equation}

As described in~\cite{zhao2017new} and~\cite{wieder2013spin}, $E_i$ under FIFO spin locks is bounded by Eq.~\ref{eq:e}. 
During the release of $\tau_i$, the number of requests for $r^k$ issued by $\tau_i \cup \text{lhp}(i)$ is $\zeta^k_i = N_i^k + \sum_{\tau_h\in{\text{lhp}(i)}} \left \lceil \frac{R_i}{T_h} \right \rceil \cdot N_h^k $, and the number of requests from $\Gamma_{m}$ is computed by $\xi^k_{i,m} = \sum_{\tau_j \in \Gamma_{m}} \left \lceil \frac{R_i + R_j}{T_j} \right \rceil \cdot N_j^k$ with the back-to-back hit taken into account~\cite{wieder2013spin}. 
For a request to $r^k$ issued by $\tau_i$ or $\tau_h \in \text{lhp}(i)$, it can be blocked by at most one request to $r^k$ on a remote processor $\lambda_m$ (if it exists).
Thus, by taking $\min\{\zeta^k_i, \xi^k_{i,m}\}$ on each $\lambda_m$, $E_i$ is computed by the exact number of remote requests that can cause the spin delay of $\tau_i$ in the worst case. This ensures a remote request is accounted for only once in $E_i$~\cite{zhao2018thesis, zhao2020priority}. 
\begin{equation}\label{eq:e}
\small
E_i = \sum_{r^k \in \mathbb{R}} \bigg(\sum_{\lambda_m \neq A_i} \min \{\zeta^k_i,  \xi^k_{i,m}\} \cdot c^k \bigg)
\end{equation}



The bounds on $B_i$ and $W_i$ vary in these protocols due to different resource accessing (spin and execution) priorities.
Here we highlight the major differences between the protocols. Detailed computations can be found in~\cite{zhao2018thesis} and~\cite{wieder2013spin}. 

In MSRP, $B_i$ is caused by a $\tau_l$ that requests either a local resource with a ceiling not lower than $P_i$ (with PCP applied) or a global resource.
Thus, $B_i$ is bounded by such a request that yields the highest delay on $\tau_i$, including potential remote blocking due to the non-preemptive spinning. 
For PWLP, $B_i$ is imposed from the same set of resources as MSRP due to the non-preemptive resource execution, but is reduced to one critical section only as tasks spin at their base priorities. As for MrsP, $B_i =0$ if $P_i$ is higher than the ceiling priority of any $r^k \in F(\text{llp}(i))$; otherwise, the same computation of MSRP applies. As revealed in~\cite{zhao2018thesis}, a general trend is observed that \textit{a resource accessing rule with a higher priority causes an increasing $B_i$}, where tasks in MSRP have the highest $B_i$.

As for $W_i$, tasks in MSRP do not incur any additional blocking (\ie, $W_i=0$) due to the non-preemptive approach. However, for both PWLP and MrsP, the bound on $W_i$ correlates to the number of preemptions (NoP) that $\tau_i$ can incur when it is accessing a resource. As described in~\cite{zhao2018thesis}, $W_i$ follows the general trend that \textit{a lower spin priority leads to an increasing $W_i$}. In particular, $W_i$ can become significant under PWLP if NoP is high~\cite{zhao2018thesis}. In addition, $W_i$ under MrsP is also sensitive to the migration cost and can become non-trivial when $\tau_i$ experiences frequent migrations~\cite{zhao2017new}.









Based on the above, tasks under different FIFO spin-based protocols exhibit various bounds on $B_i$ and $W_i$ due to different resource accessing priorities. 
The non-preemptive approach (\ie, MSRP) eliminates $W_i$ but imposes the largest $B_i$ compared to others, and $B_i$ increases for tasks with a higher priority.
In contrast, protocols with preemptable spinning (\ie, PWLP and MrsP) reduce $B_i$, however, they impose the additional blocking $W_i$ on the preempted spinning tasks. 
As shown in~\cite{wieder2013efficient,zhao2018thesis}, the performance of these protocols is sensitive to the characteristics of resources and the requesting tasks. For example, MSRP is not suitable for long resources, whereas PWLP and MrsP are not favourable if a task incurs frequent preemptions when accessing a resource~\cite{zhao2018thesis}.

\textbf{Motivation.} 
In practice, tasks often have various demands on shared resources. However, existing spin-based solutions either enforce rigid spin priority rules (\eg, non-preemptively and base priority) or allow suspension behaviours with overly pessimistic analysis, leading to significant blocking or resource over-provisioning that compromises the effectiveness of these solutions. 
To address this, we present a novel spin-based protocol named \pname{} with an accurate timing analysis that enables and exploits \textit{flexible spin priority} to improve the timing performance of FP-FPS systems with shared resources.

\section{\pname{}: Working Mechanism and Properties} 
\label{sec:protocol}

This section describes the basic working mechanism of \pname{}. Unlike existing protocols, \pname{} enables flexible spin priority in which a task can spin at any priority within a given range for accessing a resource. With the working mechanism, a set of properties held by \pname{} is presented, demonstrating the predictable worst-case resource accessing behaviour of tasks under \pname{}. 
This provides the foundation of the blocking analysis and the spin priority assignment constructed in Sec~\ref{sec:RTA} and Sec~\ref{sec:assignment}, respectively.



\textbf{Working Mechanism.}
As with the existing spin-based solutions, \pname{} manages global resources using FIFO spin locks, and controls local ones by PCP. 
In addition, the following resource accessing rules are defined for \pname{}, specifying the behaviour of $\tau_i$ for accessing a global resource $r^k$.
\begin{enumerate}
[label=R\arabic*., ref=\arabic*]
\item \label{rule1} If $\tau_i$ requests $r^k$ that is locked by another task, $\tau_i$ enters into the FIFO queue and actively spins for $r^k$ on $A_i$, with a spin priority of $P_i^k$; 
\item \label{rule2} When $\tau_i$ is granted with $r^k$ (\ie, at the head of the queue), it executes non-preemptively until $r^k$ is released, at which $\tau_i$ exits the queue and restores its priority to $P_i$; and
\item \label{rule3} If $\tau_i$ is preempted when spinning, $\tau_i$ cancels the request, exits from the FIFO queue, and is then switched away by the scheduler with its priority restored to $P_i$; once $\tau_i$ is resumed, it re-requests $r^k$ again at the end of the queue and spins with the spinning priority $P_i^k$. 
\end{enumerate}

Rule 1 provides the feature of flexible spin priority, in which each $P_i^k$ in the system can be different.
In addition, we prohibit preemptions during a critical section in Rule 2 to avoid the delay from $\tau_h$ on other tasks in the FIFO queue.
To handle preempted spinning tasks, the request cancellation from PWLP is applied in Rule 3. 
However, migrations (see MrsP in Sec.~\ref{sec:related-protocol}) are not considered to avoid additional demand for the operating system and complicated implementations.



In addition, Constraint~\ref{cons:sp} is applied to specify the range of $P_i^k$. For $\tau_i$, we enforce $P_i \leq P_i^k, \forall r^k \in F(\tau_i)$ so that $\tau_i$ cannot be preempted by any $\tau_l$ when it is 
spinning with any $P_i^k$. 
Notation $\widehat{P}$ denotes priority that allows $\tau_i$ to execute non-preemptively, \eg, the highest priority of all tasks. 
This indicates the application of \pname{} does not require additional priority levels beyond the base priorities.

\begin{constraint} \label{cons:sp}
$ \forall \tau_i \in \Gamma, \forall r^k \in F(\tau_i): P_i \leq P_i^k \leq \widehat{P}$.
\end{constraint}

With flexible spinning, each task can be assigned an appropriate spin priority within the range for accessing each resource it requests, so that the resulting blocking of tasks can be managed based on their execution urgency. 
For $\tau_i$ with a tight deadline, we can assign $P_i^k = \widehat{P}, \forall r^k \in F(\tau_i)$ and $P^k_l < P_i, \forall r^k \in F(\tau_l)$ so that $W_i = 0$ with a minimised $B_i$.

\vspace{5pt}
\textbf{Properties of \pname{}.} 
With Rules 1-3 and Constraint~\ref{cons:sp} applied, the following properties (represented as lemmas) hold for \pname{}. These properties justify the feasibility of bounding the worst-case response time of tasks ruled by \pname{}.
\begin{lemma} \label{lem:interference}
\pname{} is compliant with the interference bound of FP-FPS system, \ie, $I_i=\sum_{\tau_h\in{\text{lhp}(i)}} \left\lceil{\frac{R_{i}}{T_h}}\right\rceil \cdot {\overline{C_h}}$.
\end{lemma}
\begin{proof}
This follows directly from Constraint~\ref{cons:sp}. For $\tau_i$ and $\tau_l \in \text{llp}(i)$, it is guaranteed that $P_l < P_i \leq P_i^k, \forall r^k \in F(\tau_i)$. Thus, $\tau_i$ cannot be preempted by any newly-released $\tau_l$ (with a priority of $P_l$) under an FP-FPS system, hence, the lemma.
\end{proof}

\begin{lemma} \label{lem:onerequest}
At most one task per processor can request a resource at a time instant.
\end{lemma}
\begin{proof}
This is ensured by Rules 2 and 3. With Rule 2, $\tau_i$ cannot be preempted when executing with a resource. In addition, Rule 3 guarantees that $\tau_i$ always cancels its request if it is preempted during spinning. Thus, the lemma holds.
\end{proof}

\begin{lemma} \label{lem:fb}
Upon $\tau_i$'s arrival, it can be blocked at most once due to resource requests issued by $\tau_l \in \text{llp}(i)$.
\end{lemma}
\begin{proof}
This follows directly from Lemma~\ref{lem:onerequest}. At a time instant, at most one task on a processor can have its priority raised for accessing a resource, whereas others wait in the ready queue with a base priority. Therefore, when $\tau_i$ arrives, it can be blocked by at most one request issued by $\tau_l$.
\end{proof}

\begin{lemma} \label{lem:rerequests}
If $P_h > P_i^k$, $\tau_i$ can issue at most $\left\lceil{\frac{R_{i}}{T_h}}\right\rceil$ re-requests for $r^k$ due to the preemptions of $\tau_h$, with cancellation applied.
\end{lemma}
\begin{proof}
With Rule 3, $\tau_i$ issues a re-request for $r^k$ each time it is preempted during spinning. For $\tau_h$ with $P_h > P_i^k$, it imposes at most $\left\lceil{\frac{R_{i}}{T_h}}\right\rceil$ preemptions on $\tau_i$, hence, the lemma holds.
\end{proof}

Lemma~\ref{lem:interference} provides the bound of $I_i$ directly, and Lemmas~\ref{lem:onerequest} to~\ref{lem:rerequests} demonstrate the feasibility for bounding $E_i$, $B_i$ and $W_i$ under \pname{}. 
Based on the above, we show that \pname{} can provide fine-grained resource control via the flexible spin priorities, with predictable worst-case resource accessing behaviour of tasks.
However, the application of flexible spinning raises two major challenges:
\begin{itemize}
\item How to compute the worst-case blocking of $\tau_i$ with the flexible spin priority applied?
\item How to assign effective task spin priority given that it has a direct impact on the blocking of tasks?
\end{itemize}

In the following sections, we tackle the challenges by proposing (i) a new analysing technique that provides the worst-case blocking bound of tasks under \pname{} (Sec.~\ref{sec:RTA}); and (ii) a spin priority assignment (Sec.~\ref{sec:assignment}) that determines each $P_i^k$ in the system, providing predictability guarantee and enhancing the timing performance of \pname{}, respectively. 


\begin{figure}[t]
\centering
\hspace{-5pt}\includegraphics[width=.95\columnwidth]{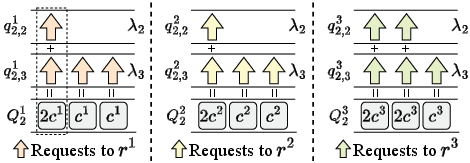}
\caption{Construction of $Q_2^k$ for $\tau_2$.}
\label{fig:construct_BQ}
\vspace{-10pt}
\end{figure}

\begin{figure*}[t]
\centering
\subfigcapskip=5pt 
\subfigure[The computation of $E_2$.]{\label{fig:example_E}  
{\includegraphics[width=0.5\columnwidth]{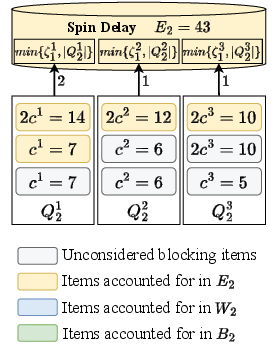}}}
\subfigure[Compute $B_2+W_2$ in $B_2 \rightarrow W_2$.]{\label{fig:example_BW}  
\includegraphics[width=.48\columnwidth]{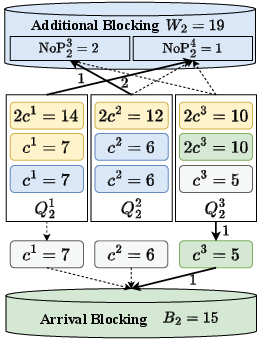}}
\subfigure[Compute $B_2+W_2$ in $W_2 \rightarrow B_2$.]{\label{fig:example_WB}  
{\includegraphics[width=.48\columnwidth]{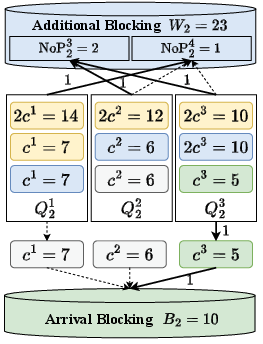}}}
\subfigure[The worst case of $B_2+W_2$.]{\label{fig:example_WorstCase}  
\includegraphics[width=.48\columnwidth]{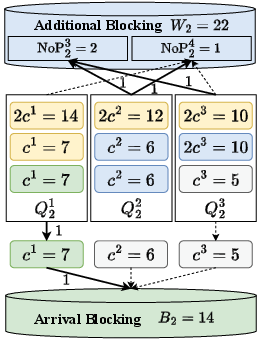}}
\caption{The problem of bounding $B_2$ and $W_2$ for $\tau_2$ in Tab.~\ref{tab:example_system}.}
\label{fig-2}
\vspace{-10pt}
\end{figure*}

\section{Blocking Analysis of \pname{}} 
\label{sec:RTA}

This section presents the blocking analysis for tasks managed by \pname{}. As with other FIFO spin-based protocols, the overall response time $R_i$ of $\tau_i$ under \pname{} is given in Eq.~\ref{eq:r}, with $E_i$, $B_i$ and $W_i$ demanding a bound. 
First, with Lemma~\ref{lem:onerequest}, $E_i$ is bounded by Eq.~\ref{eq:e} as proved in~\cite{wieder2013spin} for FIFO spin locks. 
However, due to the flexible spinning, existing techniques that analyse $B_i$ and $W_i$ separately~\cite{zhao2018thesis} are no longer applicable for \pname{}.
The fundamental reason is in \pname{}, a remote request of $r^k$ can impose a blocking on either $B_i$ or $W_i$, \eg, with $P_l^k \geq P_i$ (causes $B_i$) and $P_i^k < P_h$ (causes $W_i$).
However, this cannot occur in existing protocols, as $W_i=0$ in MSRP whereas $B_i$ in PWLP do not include any remote blocking~\cite{wieder2013spin}.


Below we first illustrate the problem of bounding $B_i$ and $W_i$ under \pname{} (Sec.~\ref{sec:problemofBW}). 
Then, we identify the possible blocking sources of $B_i$ and $W_i$ (Sec.~\ref{sec:identifySource}), and construct a minimum cost maximum flow-based analysis (Sec.~\ref{sec:MCMF}) that bounds $B_i + W_i$.
To facilitate the presentation, the following definitions are introduced. 
Notations introduced in this section are summarised in Tab.~\ref{tab:list_of_symbols_in_analysis}.
An example system (see Tab.~\ref{tab:example_system}) is applied to illustrate the analysing process.

\begin{definition}\label{def:RQ}
A \textit{Request Queue} contains the execution time of requests for $r^k$ on a remote processor $\lambda_m$ during $\tau_i$'s release, represented as $q_{i,m}^k \triangleq \{c^{k},...,c^{k} \}$ with $ |q_{i,m}^k| = \xi_{i,m}^k$.
\end{definition}

\begin{definition}\label{def:BQ}
A \textit{Blocking Queue} provides a list of blocking that requests from $A_i$ to $r^k$ can incur, denoted as $Q^k_i \triangleq
\bigcup_{n=1}^{\kappa}
\{\sum_{\lambda_m \neq A_i} q_{i,m}^k(n)\}$ with $\kappa = max\{ \xi_{i,m}^k ~|~ \lambda_m \neq A_i \}$. 
\end{definition}

The $n$\textsuperscript{th} item in $Q^k_i$ is obtained by adding the $n$\textsuperscript{th} element in each $q_{i,m}^k$ (if it exists). The length of $Q^k_i$ (\ie, $|Q^k_i|$) indicates the number of requests from $A_i$ to $r^k$ that can incur blocking, and values in $Q^k_i$ are the worst-case blocking time. 
Fig.~\ref{fig:construct_BQ} illustrates the construction of $q^k_{2,m}$ and $Q^k_2$ for $\tau_2$ in the example system (Tab.~\ref{tab:example_system}). The requests issued for each resource from $\lambda_2$ and $\lambda_3$ during $R_2$ are shown in the figure.

\begin{table}[t]
\caption{Notations applied in the proposed analysis.}
\label{tab:list_of_symbols_in_analysis}
\centering
\begin{tabular}{p{.13\columnwidth}p{.77\columnwidth}}
\hline
\textbf{Notations} & \textbf{Descriptions} \\
\hline
$\text{NoP}_{i}^h$ & Number of preemptions that $\tau_i$ incurs from $\tau_h$ during $R_i$. \\
$q_{i,m}^k$ & A list of resource execution time of requests to $r^k$ issued from tasks on a remote processor $\lambda_m$ during $R_i$. \\
$Q_{i}^k$ & A list of worst-case remote blocking incurred by requests to $r^k$ issued from $A_i$ during $R_i$. \\
$\alpha_{i}^k$ & The starting index (if it exists) of the unaccounted blocking items in $Q_{i}^k$ after the computation of $E_i$. \\
\hline
$\mathcal{L}b_i$ & A list of blocking items that can be accounted for in $B_i$. \\
$\mathcal{L}w_{i,h}$  & A list of blocking items that can be accounted for in $W_i$ due to $\tau_h \in \text{lhp}(i)$.  \\
$F^{b}(\tau_i)$ & Resources that can cause the arrival blocking to $\tau_i$. \\
$F^{w}(\tau_i,\tau_h)$ & Resources that can impose additional blocking due to $\tau_h$.\\
\hline
\end{tabular}
\vspace{-10pt}
\label{tab:list_of_symbols}
\end{table}

\subsection{The Problem of Bounding $B_i$ and $W_i$}\label{sec:problemofBW}

This section illustrates that existing analysing approaches can lead to incorrect bounds for \pname{}, with the analysis of $\tau_2$ as an example. 
We assume $\tau_3$ can release twice and $\tau_4$ can release once during the release of $\tau_2$, denoted as $\text{NoP}_2^3 = 2$, $\text{NoP}_2^4 = 1$. The spin priority of $\tau_2$ for each resource is $P^1_2=3$, $P^2_2=2$ and $P^3_2=2$, as shown in Tab.~\ref{tab:example_system}. 
Based on Eq.~\ref{eq:e}, $\tau_2$ incurs a spin delay of $E_2 = 43$. The blocking items in $Q_2^k$ that are accounted for in $E_2$ are highlighted in Fig.~\ref{fig:example_E}. 

Because $P_3 > P^2_2 = P^3_2$, $\tau_3$ can preempt $\tau_2$ twice in total (\ie, $\text{NoP}_2^3 = 2$) when $\tau_2$ is spinning for $r^2$ or $r^3$. In addition, $\tau_4$ can preempt $\tau_2$ once (\ie, $\text{NoP}_2^4 = 1$) when $\tau_2$ is waiting for any of the resources. 
Therefore, in the worst case, $\tau_2$ can issue three re-requests due to the cancellation (two for $r^2$ or $r^3$ and one for any resource), each can incur an additional blocking imposed by the remote requests. 
Below we illustrate the detailed computation of $B_2$ and $W_2$ based on the remote blocking items in $Q_2^k$ that are not accounted for in $E_2$.
Examples~\ref{eg:BW} and~\ref{eg:WB} compute $B_2$ and $W_2$ separately (as with existing analysis) in different orders, and Example~\ref{eg:worst} provides the actual worst-case bound on $B_2+W_2$.



\begin{example} \label{eg:BW}
\textit{(Bounding $B_2$ before $W_2$.)} 
\normalfont
The computation is illustrated in Fig.~\ref{fig:example_BW}. 
First, as $P_1^1=3$ and $P_1^3=2$, $\tau_2$ can incur $B_2$ with remote blocking due to one access of $\tau_1$ to either $r^1$ or $r^3$.
The worst-case of $B_2$ occurs when $\tau_2$ is blocked by $\tau_1$'s access to $r^3$, \ie, $B_2 = 10 + 5 = 15$, including one resource execution of $\tau_1$ itself. 
Then, $W_2$ is computed by the blocking items that are not accounted for in $E_2$ and $B_2$, leading to the highest bound of $W_2 = 7 + 6 + 6 = 19$, \ie, $\tau_2$ re-requests $r^1$ once and $r^2$ twice.
Thus, by computing $B_2$ and $W_2$ in order, we have a total blocking of $B_2+W_2 = 34$.
\end{example}

\vspace{-8pt}

\begin{example} \label{eg:WB}
\textit{(Bounding $W_2$ before $B_2$.)} 
\normalfont
The computation is illustrated in Fig.~\ref{fig:example_WB}.
First, the worst-case bound on $W_2$ is computed by $W_2 = 7 + 6 + 10 = 23$, indicating $\tau_2$ re-requests each resource once due to the cancellation.
Then, the highest $B_2$ is imposed by $r^3$ with $B_2 = 5+5$, where all blocking items for $r^1$ have been accounted for.
Therefore, by computing $W_2$ before $B_2$, we obtained a total blocking of $W_2 + B_2 = 33$.
\end{example}

\begin{table}[t]
\caption{An example system with 3 processors, 6 tasks, and 3 resources \textit{($\tau_2$ is the currently-examined task)}. }
\label{tab:example_system}
\centering
\resizebox{\columnwidth}{!}{
\begin{tabular}{c|c|c|c|c|c}
 \hline
 $\tau_x$ & $A_x$ & $P_x$ & $F(\tau_x)$ & $N_x^1,N_x^2,N_x^3$ & $P_x^1,P_x^2,P_x^3$ \\
 \hline
 $\tau_1$   & \multirow{4}{*}{$\boldsymbol{\lambda_1}$} & 1 & $\{ r^1, r^2, r^3 \}$  & 1, 1, 1  & 3, 1, 2 \\
 $\boldsymbol{\tau_2}$  &  & \textbf{2} & $\{\boldsymbol{r^1}, \boldsymbol{r^2}, \boldsymbol{r^3}\}$   & \textbf{1, 1, 1}   & \textbf{3, 2, 2} \\
 $\tau_3$  &  & 3 & $\varnothing$   & -, -, - & -, -, - \\ 
 $\tau_4$  &  & 4 & $\{ r^1\}$  & 1, -, - & 4, -, - \\ 
 \hline
 $\tau_5$   & $\lambda_2$ & 1 & $\{ r^1, r^2, r^3 \}$ & 1, 1, 2 &  1, 1, 1\\ 
 $\tau_6$   & $\lambda_3$ & 1 & $\{ r^1, r^2, r^3 \}$ & 3, 3, 3 &  1, 1, 1 \\ 
 \hline 
 \hline 
 $r^k$ &  \multicolumn{2}{c|}{Resource Type} & $c^k$ & \multicolumn{2}{c}{Tasks that request $r^k$}\\
 \hline 
 $r^1$ &  \multicolumn{2}{c|}{Global} & 7 & \multicolumn{2}{c}{$\{\tau_1, \tau_2, \tau_4, \tau_5, \tau_6\}$}\\
 $r^2$ &  \multicolumn{2}{c|}{Global} & 6 & \multicolumn{2}{c}{$\{\tau_1, \tau_2, \tau_5, \tau_6\}$}\\
 $r^3$ &  \multicolumn{2}{c|}{Global} & 5 & \multicolumn{2}{c}{$\{\tau_1, \tau_2, \tau_5, \tau_6\}$}\\
 \hline
\end{tabular}}
\vspace{-10pt}
\end{table}


\vspace{-10pt}

\begin{example} \label{eg:worst}
\textit{(The worst case of $B_2 + W_2$.)}
\normalfont
The computation is illustrated in Fig.~\ref{fig:example_WorstCase}. 
By examining the $Q_2^k$, we can observe a case in which $\tau_2$ incurs $B_2$ from $r^1$ while incurs $W_2$ from $r^2$ twice and $r^3$ once. In this case, $B_2 = 7 + 7 = 14$ and $W_2 = 6 + 6 + 10 = 22$, leading to $B_2 + W_2 = 36$. This reflects the worst-case blocking that $\tau_2$ can incur, providing evidence that existing analysing approaches can lead to incorrect results. 
\end{example}

As shown in the example, $\tau_2$ can incur a remote blocking in $B_2$ from $\{r_1, r_3\}$ and in $W_2$ from $\{r_1, r_2, r_3\}$, which are computed based on blocking items from the same blocking queues, \ie, $Q_2^1$ and $Q_2^3$.
That is, computing $B_2$ and $W_2$ independently in any order cannot yield the worst-case blocking time, in which $B_2$ (resp. $W_2$) may take the highest blocking item available in the queue, ignoring the fact that it can cause a decrease in $W_2$ (resp. $B_2$). 
Therefore, to bound the worst-case blocking in \pname{}, new analysing techniques are required that compute $B_i$ and $W_i$ in a collaborative fashion.

\subsection{Identifying Sources of $B_i$ and $W_i$}\label{sec:identifySource}
To bound $B_i$ and $W_i$ collaboratively, we first identify the blocking items in $Q_i^k$ that can be accounted for in $B_i$ and $W_i$, and then compute the bounding of $B_i+W_i$ based on these identified blocking items (see Sec.~\ref{sec:MCMF}).

\textbf{Blocking items of $B_i$.}
Under \pname{}, the set of resources that can cause $B_i$ is obtained in Eq.~\ref{eq:Fb}, in which $r^k \in F(\text{llp}(i))$ can impose a blocking if it is a local resource with a ceiling (denoted as $\widehat{P}^k_{\lambda_m}$) not lower than $P_i$  or is globally shared.
\begin{equation}\label{eq:Fb}
\small
F^b(\tau_i) = \{ r^k | r^k \in F(\text{llp}(i)) \wedge (\widehat{P}^k_{\lambda_m} \geq P_i \vee r^k \text{ is global} )   \}
\end{equation}

The resources in $F^b(\tau_i)$ can be classified into two types, depending on whether it can impose a remote blocking: 
\begin{itemize}
\item ones that only cause a blocking of $c^k$ -- local resources with $\widehat{P}^k_{\lambda_m} \geq P_i$ or global ones with $P^k_l < P_i$; and
\item ones that can include the remote blocking -- global resources with $P^k_l \geq P_i$.
\end{itemize}

Accordingly, the list of blocking items of $r^k$ that could be accounted for in $B_i$ is constructed in Eq.~\ref{eq:lb}, denoted as $\mathcal{L}b^k_i$. The notation $\alpha_i^k = min\{\zeta_i^k, |Q_{i}^k| \} + 1$ provides the starting index of the unaccounted blocking items in $Q_{i}^k$ (if it exists) after the computation of $E_i$.
\begin{equation} \label{eq:lb}
\small
\mathcal{L}b^k_i = 
\begin{cases} 
\{ c^k \}, ~~\text{if } \widehat{P}_{\lambda_m}^k \geq P_i \vee (r^k \text{ is global }  \wedge P_l^k < P_i) \\
\{ c^k \} \cup \bigcup_{n=\alpha_i^k}^{|Q_{i}^k|} \big\{ c^k + Q^k_i(n) \big\} , ~~~~~~~~~~~~~~~~~\text{else}\\
\end{cases}
\end{equation}

Note, if $r^k$ belongs to the second type, it is still possible that $r^k$ only imposes a blocking of $c^k$. This can happen if all items in $Q_i^k$ are accounted for in $E_i$ and $W_i$ during the computation, \eg, the case of $Q_2^1$ in Fig.~\ref{fig:example_WB}. Example~\ref{eg:lb} illustrates the construction of $\mathcal{L}b_2^k$ of $\tau_2$ in Tab.~\ref{tab:example_system}.

\begin{example} \label{eg:lb}
\normalfont
For $\tau_2$ in the example, $F^b(\tau_2)=\{r^1,r^2,r^3\}$ with $r^2$ included due to the non-preemptive execution. 
Thus, we have $\mathcal{L}b^1_2 = \{c^1, c^1 + c^1\}$, $\mathcal{L}b^2_2 = \{c^2\}$,  $\mathcal{L}b^3_2 = \{c^3, c^3 + 2c^3, c^3 + c^3\}$, \ie, a local blocking only or with remote blocking based on the unaccounted items (see Fig.~\ref{fig:example_E}).  
\end{example}


Finally, as described in Lemma~\ref{lem:fb}, the arrival blocking of $\tau_i$ can occur at most once by a request of $\tau_l$. 
Therefore, $B_i$ is obtained by taking at most one item from $\mathcal{L}b_i$ (Eq.~\ref{eq:BBB}), which contains all the blocking items of $\mathcal{L}b_i^k$ for all $r^k \in F^b(\tau_i)$.
\begin{equation}\label{eq:BBB}
\small
\mathcal{L}b_i = \bigcup_{r^k \in F^b(\tau_i)} \mathcal{L}b^k_i
\end{equation}

\textbf{Blocking items of $W_i$.}
Under \pname{}, a $\tau_h$ can preempt the spinning of tasks in $\tau_i \cup \text{lhp}(i)$ if it has a higher priority. This causes resource re-requests that can impose a (transitive) blocking on $W_i$, if there exist unaccounted remote blocking items in $Q_i^k$ (see Fig.~\ref{fig-2}).
The set of resources that can impose $W_i$ on $\tau_i$ due to $\tau_h$'s preemption is identified by Eq.~\ref{eq:Fw}.
\begin{equation} \label{eq:Fw}
\small
F^w(\tau_i, \tau_h) = \{ r^k ~\big|~   P_h > P_x^k,   \tau_x \in \tau_i \cup \text{lhp}(i) \}
\end{equation}

With $F^{w}(\tau_i,\tau_h)$ constructed, the candidate blocking items of $W_i$ caused by the preemptions of $\tau_h$ are obtained by Eq.~\ref{eq:lw}, denoted as $\mathcal{L}w_{i,h}$. For each $r^k \in F^w(\tau_i,\tau_h)$, the items in $ Q^k_i$ that are not accounted for in $E_i$ could be considered (if they exist). Accordingly, $n$ starts at $\alpha_i^k$ as shown in the equation.
\begin{equation}\label{eq:lw}
\small
\mathcal{L}w_{i,h} = \bigcup_{r^k \in F^w(\tau_i,\tau_h)} \bigcup_{n=\alpha_i^k}^{|Q_{i}^k|} \Big\{ Q^k_i(n) \Big\}
\end{equation}

In the worst case, $\tau_i$ can incur $\text{NoP}_i^h = \left \lceil \frac{R_i}{T_h} \right \rceil$ preemptions from $\tau_h$, and each preemption can cause an additional delay (Lemma~\ref{lem:rerequests}). Thus, for a given $\tau_h$, at most $\text{NoP}_i^h$ items can be accounted for in $\mathcal{L}w_{i,h}$ (if they exist) to bound the worst-case $W_i$. 
Example~\ref{eg:lw} shows the construction of $\mathcal{L}w_{2,3}$ and $\mathcal{L}w_{2,4}$.

\begin{example} \label{eg:lw}
\normalfont
For $\tau_2$ in the example, $F^w(\tau_2,\tau_3)=\{r^2, r^3\}$ and $F^w(\tau_2,\tau_4) = \{r^1, r^2, r^3\}$ based on Eq.~\ref{eq:Fw}. Then, we have $\mathcal{L}w_{2,3} = \{c^2, c^2, 2c^3, c^3\}$ and $\mathcal{L}w_{2,4} = \{c^1, c^2, c^2, 2c^3, c^3\}$, \ie, unconsidered blocking items in Fig.~\ref{fig:example_E}.
\end{example}

It is important to note that for two local high-priority tasks (say $\tau_{a}$ and $\tau_b$), they can preempt the accesses for the same $r^k$ if $F^w(\tau_i, \tau_a) \cap F^w(\tau_i, \tau_b) \neq \varnothing$, \eg, $r^2$ and $r^3$ in Example 5. Thus, $\mathcal{L}w_{i,a}$ and $\mathcal{L}w_{i,b}$ share the same blocking items in a $Q_i^k$ that can only be accounted for once. 
In addition, 
the same can occur between $\mathcal{L}b_i$ and $\mathcal{L}w_{i,h}$, where both of them contain the blocking items of $r^k \in F^b(\tau_i) \cap F^w(\tau_i,\tau_h)$, \eg, the blocking items of $r^1$ and $r^3$ as shown in Example~\ref{eg:lb} and~\ref{eg:lw}. 
As described, this is caused by the flexible spinning in \pname{}. Under \pname{}, $\tau_i$ can be blocked by $r^k$ upon its arrival, and incurs additional blocking from the same resource if $\tau_i$ is preempted when accessing $r^k$.
Therefore, Constraint~\ref{cons:problem} is constructed to guarantee a blocking item in $Q_i^k$ is not considered multiple times during the computation.  


\begin{constraint} \label{cons:problem}
The shared blocking items in 
$\mathcal{L}b_i$ and $\mathcal{L}w_{i,h}$ can be accounted for at most once when bounding $B_i+W_i$. 
\end{constraint}

\subsection{Bounding $B_i+W_i$ via MCMF}\label{sec:MCMF}

\textbf{Problem Formulation.}
Based on the constructed $\mathcal{L}b_i$ and $\mathcal{L}w_{i,h}$, we have demonstrated that the worst-case blocking bound of $B_i+W_i$ in nature is an optimisation problem~\cite{burkard2012assignment}, as formulated below. Notations $b_i$ and $w_{i,h}$ are the set of blocking items accounted for in $B_i$ and $W_i$, respectively.
The objective is to maximise $B_i+W_i$ with the constraints enforced. 
\begin{equation*} \label{eq:formualtion}
\small
\begin{aligned}
& \text{Given:} & & \Gamma, ~\Lambda, ~\mathbb{R} \\
& \text{Maximise}
& & B_i + W_i = \sum b_i + \sum_{\tau_h \in \text{lhp}(i)} \sum w_{i,h}  \\  
& \text{On} & & b_i \subseteq \mathcal{L}b_i,~ w_{i,h} \subseteq \mathcal{L}w_{i,h}\\
& \text{s.t.} & & |b_i| \leq 1,~|w_{i,h}| \leq \text{NoP}_i^h, \text{ and Constraint~\ref{cons:problem}} \\
\end{aligned}
\end{equation*}
\normalsize

Intuitively, this problem can be solved by the Mixed Integer Linear Programming (MILP)~\cite{wieder2013spin}. However, the MILP suffers from scalability issues with a high computation cost, imposing application difficulties in practice~\cite{glover1974implementation, bixby1980converting, zhao2017new}. 
Therefore, we construct an analysis of $B_i+W_i$ based on the minimum-cost maximum-flow problem~\cite{ford1956maximal,schrijver2002history}, providing a cost-effective approach for bounding the blocking under \pname{}.

\textbf{MCMF network.} 
The MCMF is extensively applied in the field of optimisation~\cite{magnanti1993network,cruz2023survey}. For a given problem, MCMF constructs a flow network and optimises towards the objective, by identifying the way to send the maximum flow through the network with the lowest cost~\cite{ahuja1995network}.
An \textit{MCMF network} is constructed as a directed graph $\mathcal{G} = (\mathsf{V}, \mathsf{E})$ containing a set of nodes $\mathsf{V}$ and edges $\mathsf{E}$~\cite{magnanti1993network}.
The network $\mathcal{G}$ has one source node $v_{src}$ and one sink node $v_{snk}$.
An edge connecting two nodes $(v_x, v_u) \in \mathsf{E}$ has a capacity of $\kappa_{x}^{u}$ and a cost of $\delta_{x}^{u}$, indicating the maximum flow allowed on this edge and the associated cost. 
A flow goes through edge $(v_x, v_u)$ has a volume of $f_{x}^{u}$ that must satisfy $0 \leq f_{x}^{u} \leq \kappa_{x}^{u}$ with a cost of $f_{x}^{u} \times \delta_{x}^{u}$.
Except for $v_{src}$ and $v_{snk}$, the total incoming and outgoing flow of a node must be identical. 
Network $\mathcal{G}$ has a flow of $\mathsf{F}$ going through the whole network from $v_{src}$ to $v_{snk}$.

In our context, a node represents (i) a blocking effect (\eg, the arrival blocking or the additional blocking caused by a $\tau_h$), or (ii) a blocking item in $Q_i^k$ that is not yet accounted for. Edges from a blocking effect to the items in $Q_i^k$ are created based on $\mathcal{L}b_i$ and $\mathcal{L}w_{i,h}$, indicating the items can be accounted for in that blocking effect. The capacity on an edge is the number of requests that can impose blocking, whereas the cost gives the actual blocking time. The objective is to obtain the maximum blocking (\ie, cost) in the network based on the number of (re-)requests that cause blocking in $B_i+W_i$. 
We note that the maximum cost can be effectively obtained with a trivial modification that sets all costs in $\mathcal{G}$ as negative.




\vspace{-3pt}
\SetKwInOut{Initialise}{Initialise}
\begin{algorithm}[t]
\small
\Initialise{$\mathsf{V} = \{v_{src}$, $v_{snk}\}$, $\mathsf{E}=\{\}$}

{\scriptsize\ttfamily/* {$\blackcircled{a}$ Constructing nodes for blocking items.} */} \\
\For{{\normalfont each $r^k \in F^b(\tau_i) \cup F^w(\tau_i)$}}{

\While{$ \alpha_i^k \leq n \leq |Q_i^k|$}{
$\mathsf{V} = \mathsf{V} \cup \{v_{x}\}$;~~
$\mathsf{E} = \mathsf{E} \cup \texttt{e}(v_{x},v_{snk}, 1, 0)$
}

}

{\scriptsize\ttfamily/* {$\blackcircled{b}$ Constructing nodes and constraints of $B_i$.}*/} \\
    $\mathsf{V} = \mathsf{V} \cup \{v_{B}\}$;~~$\mathsf{E} = \mathsf{E} \cup \texttt{e}(v_{src}, v_{B},1,0)$\\
    \For{\normalfont each $r^k \in F^b(\tau_i)$}{
        $\mathsf{V} = \mathsf{V} \cup \{v_{k}\}$;\\
        $\mathsf{E} = \mathsf{E} \cup \texttt{e}(v_{B},v_{k},1,c^k) \cup \texttt{e}(v_{k},v_{snk},1,0)$;\\
        \For{{\normalfont each $c^k + Q_i^k(n)$ in $\mathcal{L}b_i^k$}}{
        $\mathsf{E} = \mathsf{E} \cup \texttt{e}(v_{k}, \texttt{v}(Q_i^k(n)), 1, Q_i^k(n))$;\\
    }
    }

{\scriptsize\ttfamily/* {$\blackcircled{c}$ Constructing nodes and constraints of $W_i$.}*/} \\

\For{{\normalfont each $\tau_h \in \text{lhp}(i) \wedge |\mathcal{L}w_{i,h}|>0$}}{
    $\mathsf{V} = \mathsf{V} \cup \{v_{h}\}$;~~$\mathsf{E} = \mathsf{E} \cup \texttt{e}(v_{src}, v_{h}, \text{NoP}_i^h, 0)$\\
    \For{{\normalfont each $Q_i^k(n)$ in $\mathcal{L}w_{i,h}$}}{
        $\mathsf{E} = \mathsf{E} \cup \texttt{e}(v_h, \texttt{v}(Q_i^k(n)), 1, Q_i^k(n) )$;\\
    }
}


\Return $\mathcal{G}=(\mathsf{V},\mathsf{E})$;




\caption{Construction of the MCMF network.}

\label{alg:network}
\end{algorithm}

\textbf{Network Construction.} Following the above, Alg.~\ref{alg:network} presents the construction of the MCMF network (\ie, $\mathcal{G}$) for bounding $B_i+W_i$, with $v_{src}$ and $v_{snk}$ initialised.
The corresponding illustration is given in Fig.~\ref{fig:network}. To ease presentation, the number on an edge $(v_x, v_u)$ indicates its capacity (\ie, $\kappa_{x}^{u}$), and the associated cost (\ie, $\delta_{x}^{u}$) is given on node $v_u$. Notation $F^w(\tau_i) = \bigcup_{\tau_h \in \text{lhp}(i)} F^w(\tau_i, \tau_h)$ denotes all resources that can impose $W_i$.
Function $\texttt{e}(v_x, v_u, \kappa, \delta)$ constructs an directed edge from $v_x$ to $v_u$, with $\kappa_x^u=\kappa$ and $\delta_x^u = \delta$; and $\texttt{v}(Q_i^k(n))$ returns the corresponding node of a blocking item.
Essentially, the construction of $\mathcal{G}$ contains three major steps:
\begin{enumerate} [label=\arabic*.]
\item Nodes ($v_x$) for the blocking items in $Q_i^k$ are constructed, which can be accounted for in $B_i$ or $W_i$ (lines 1-6);
\item Nodes are created for $B_i$, and edges are added to connect $v_x$ based on $\mathcal{L}b_i$, indicating blocking items that can be accounted for in $B_i$ (lines 7-15); and
\item Nodes and edges for $W_i$ are constructed based on $\mathcal{L}w_{i,h}$ (lines 16-22). 
\end{enumerate}

\begin{figure}[t]
    \centering
    \vspace{-5pt}
    \includegraphics[width=\columnwidth]{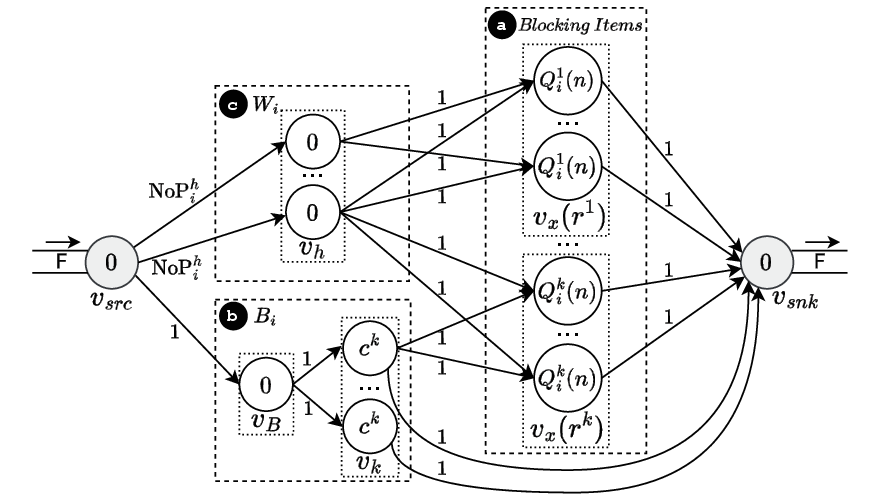}
    \caption{An illustrative MCMF network for bounding $B_i+W_i$.}
    \label{fig:network}
    \vspace{-15pt}
\end{figure}

For the first step (Fig.~\ref{fig:network}\blackcircled{a}), the algorithm iterates through each $r^k$ that can cause $B_i$ or $W_i$ (line 2), and creates a node $v_x$ for every item in $Q_i^k$ that has not being accounted for in $E_i$, \ie, $Q_i^k(n)$ with $\alpha_i^k \leq n \leq |Q_i^k|$ (line 3). The node is connected to $v_{snk}$ with a capacity of 1 and a cost of 0 (line 4), ensuring at most one flow with a volume of 1 can pass through the node. 
This ensures a blocking item is accounted for at most once even if it exists in $\mathcal{L}b_{i}$ and multiple $\mathcal{L}w_{i,h}$. Hence, Constraint~\ref{cons:problem} holds during the computation of $B_i+W_i$.


With nodes of the blocking items constructed, the second step (Fig.~\ref{fig:network}\blackcircled{b}) creates the node $v_B$ for computing $B_i$. Node $v_B$ is connected from $v_{src}$ with a capacity of 1 and a cost of 0 (line 8), which indicates at most one request to $r^k \in F^b(\tau_i)$ can impose $B_i$.
As $r^k$ can cause a local delay due to the non-preemptive execution, a node $v_k$ is constructed for every $r^k \in F^b(\tau_i)$ (line 10). Each of such nodes is connected from $v_{B}$ by an edge with a capacity of 1 and a cost of $c^k$; and is connected to $v_{snk}$ (line 11). 
In addition, if $r^k$ can impose a remote blocking (see Eq.~\ref{eq:lb}), an edge is added from $v_k$ to each of the blocking items in $\mathcal{L}b_i^k$ with a capacity of 1 and a cost of $Q_i^k(n)$ (lines 12-14). By doing so, we guarantee that at most one item in $\mathcal{L}b_{i}$ is accounted for (\ie, $|b_i|\leq 1$), as only a flow with a volume of 1 is allowed to pass through $v_B$.

For the final step (Fig.~\ref{fig:network}\blackcircled{c}), we construct nodes for $W_i$, each of which represents the additional blocking imposed by a $\tau_h \in \text{lhp}(i)$. 
For $\tau_h$ with $|\mathcal{L}w_{i,h}| > 0$, a node $v_h$ is created (line 18). Then, $v_h$ is connected from $v_{src}$ with a capacity of $\text{NoP}_i^h$ and a cost of 0 (line 18); and is connected to each item in $\mathcal{L}w_{i,h}$ with a capacity of 1 and a cost of $Q^k_i(n)$ (lines 19-21). This guarantees at most $\text{NoP}_i^h$ blocking items can be accounted for in $W_i$ with a given $\tau_h$ (\ie, $|w_{i,h}|\leq \text{NoP}^h_i$), and each blocking item is considered only once. To this end, we have constructed the complete MCMF network for bounding the worst-case $B_i + W_i$ under \pname{}, with the correctness justified in Theorem~\ref{maximum_flow}.

\begin{theorem}\label{maximum_flow}
The MCMF network $\mathcal{G}$ yields the maximum cost (\ie, worst-case bound on $B_i + W_i$) when $\mathsf{F}$ is maximised.
\end{theorem}
\begin{proof}
We prove this theorem by contradiction. Let $\mathsf{F}$ and $\Delta$ denote the maximum flow and the highest cost of $\mathcal{G}$, respectively. Assuming $\Delta$ is achieved by $\mathsf{F}'$ with $\mathsf{F}' < \mathsf{F}$, there exists at least one feasible flow from $v_{src}$ to $v_{snk}$ that can lead to a higher cost. This contradicts the assumption that $\Delta$ is the maximum value. Hence, the theorem follows.
\end{proof}

This concludes the MCMF-based analysis that produces the bound on $B_i +W_i$, and subsequently, the complete response time analysis of \pname{} with Eq.~\ref{eq:r}.
With $\mathcal{G}$ constructed, existing MCMF solvers with polynomial-time complexity~\cite{cruz2023survey} can be applied to obtain the maximum $B_i+W_i$, with the trivial modification of the costs mentioned above.
In addition, we note that more advanced solvers (\eg, the one in~\cite{chen2022maximum}) are available with near-linear time complexity. However, this is not the focus of this paper and is referred to~\cite{cruz2023survey}. 

In addition, similar to the analysis in~\cite{wieder2013spin, zhao2017new}, the response time of $\tau_i$ in our analysis depends on the response time of potentially all tasks in the system due to the back-to-back hit (\ie, the computation of $\xi_{i,m}^k$ in Eq.~\ref{eq:e}). Therefore, with an initial response time (\eg, $R_i = \overline{C_i}$), the analysis is performed in an iterative and alternative fashion that updates $R_i$ for all tasks collectively~\cite{wieder2013spin}.
In each iteration, $E_i$ is computed by Eq.~\ref{eq:e} and $B_i+W_i$ is solved via the MCMF network based on the current $R_i$. The computation finishes if a fixed $R_i$ is obtained for each task or a task has missed its deadline.



\section{Spin Priority Assignment of \pname{}} \label{sec:assignment}

This section determines the spin priority of tasks under \pname{}.
As revealed by the constructed analysis, the choice of spin priority directly affects the blocking of tasks. For instance, given $\tau_i$ and $\tau_l \in \text{llp}(i)$, a reduction in $P_l^k$ of $\tau_l$ could reduce $B_i$ of $\tau_i$, but can increase $W_l$ for $\tau_l$ (see Eq.~\ref{eq:Fb} and~\ref{eq:Fw}).
Therefore, the assignment of spin priorities inherently involves managing the trade-offs between $B_i$ and $W_l, \forall \tau_l \in \text{llp}(i)$ so that the blocking incurred by tasks with high urgency can be effectively reduced, enhancing the performance of \pname{}.

Following this intuition, a spin priority assignment is constructed to determine $P_i^k$ of each $\tau_i$ and $r^k \in F(\tau_i)$ within the range of $P_i$ 
and $\widehat{P}$ (see Constraint~\ref{cons:sp}). 
The core idea is: from the perspective of $\tau_i$, the algorithm finds the appropriate $P_l^k$ for $\tau_l \in \text{llp}(i)$ that reduces $B_i$ of $\tau_i$ while keeping $W_l$ of $\tau_l$ reasonable, such that each task in the system can meet its deadline. 
To achieve this, we initialise $P_i^k$ with a high spin priority and examine each $\tau_i$ on a processor.
For a given $\tau_i$, a linear search is applied that gradually reduces $P_l^k$ of $\tau_l \in \text{llp}(i)$ to reduce $B_i$ so that $R_i \leq D_i$. 



Ideally, the proposed assignment can be conducted based on the constructed analysis, which provides the required timing bounds, \eg, $B_i$, $W_i$, and $R_i$. 
However, during the search, a substantial amount of invocations to the analysis are required with intensive MCMF problem-solving, imposing a significant cost that greatly undermines the applicability of the proposed assignment. 
Therefore, approximations are applied to estimate the timing bounds, which avoids high computational demand while providing effective guidance for the assignment.
Below we first illustrate the process of the proposed spin priority assignment (Sec.~\ref{sec:assignment-slack}), and then present a cost-effective approximation for the required timing bounds (Sec.~\ref{sec:assignment-approx}). 




\vspace{-5pt}
\SetKwInOut{Input}{Inputs}
\SetKwInOut{Parameter}{Parameters}
\begin{algorithm}[h]
\setlength{\textfloatsep}{0.5em} 
\small

\For{{\normalfont each $\lambda_m \in \Lambda$}}{
    {\scriptsize\ttfamily/*Initialise $P_i^k$ based on Theorem~\ref{the:special_case}.*/} \\
    \For{{\normalfont each $\tau_i \in \Gamma_{m}$ and $ r^k \in F(\tau_i)$}}{
       
        \eIf{$\zeta_i^k \geq \xi_{i,m}^k, \forall \lambda_m \in \Lambda \backslash A_i$}{ 
            $P_i^k = P_i$;
        }{
            $P_i^k = \widehat{P}$;
        }
        }

    {\scriptsize\ttfamily/*Assignment of $P_i^k$ using a linear search.*/} \\
    \For{{\normalfont each $\tau_i \in \Gamma_{m}$, highest $P_i$ first}}{
        $F^*(\tau_i) = \{ r^k~|~r^k \in F(\text{\normalfont llp}(i)) \wedge P^k_l \geq P_i \}$;\\
        
        \While{$R_i > D_i \wedge F^*(\tau_i) \neq \varnothing$}{
            $r^k \gets maxk(B_i)$;\\
            \eIf{$r^k \in F^*(\tau_i)$}{
                $P^k_l = P_i - 1, \forall \tau_l \in \text{llp}(i)$;\\
                $F^*(\tau_i) = F^*(\tau_i)~\backslash~ r^k$;\\
            } 
            {
                break;\\
            }
            
        }
    }

}

\Return $P^k_i, \forall \tau_i \in \Gamma, \forall r^k \in F(\tau_i)$;

\caption{The process of assigning spin priority.}
\label{alg:assign}
\end{algorithm}

\subsection{The Process of Spin Priority Assignment} \label{sec:assignment-slack}

Alg.~\ref{alg:assign} presents the spin priority assignment for \pname{} based on Constraint~\ref{cons:sp} that specifies the range of spin priorities of a task.
As shown at line 1, the algorithm iterates through each $\lambda_m \in \Lambda$ and assigns spin priorities of tasks executing on $\lambda_m$. 

\textbf{Initialisation.}
From lines 3 to 9, the spin priorities are initialised for the search-based assignment to take place in the next phase. To enforce a one-way linear search, the initialisation generally sets $P_i^k$ with the highest priority ($P_i^k=\widehat{P}$ at line 7), so that $P_i^k$ can only decrease during the search-based assignment. 
However, some tasks are initialised with $P^k_i=P_i$ at line 5. 
Such initialisation is conducted if all remote blocking from $r^k$ is accounted for in $E_i$, indicating that $r^k$ cannot impose any extra blocking on $\tau_i$ regardless of $P^k_i$. 
Therefore, with $P^k_i=P_i$, no extra blocking is imposed on $\tau_i$, and the arrival blocking incurred by $\text{lhp}(i)$ due to $\tau_i$'s access for $r^k$ is minimised. 
The condition at line 4 is used to identify these tasks and is justified in Theorem~\ref{the:special_case}, where $\zeta_i^k$ and $\xi_{i,m}^k$ are the number of requests issued by $\tau_i \cup \text{lhp}(i)$ and $\Gamma_m$ during the release of $\tau_i$, respectively.

\begin{theorem} \label{the:special_case}
If $\zeta_i^k \geq \xi_{i,m}^k, \forall \lambda_m \in \Lambda \backslash A_i$, $\tau_i$ incurs the same worst-case blocking caused by $r^k$ with any $P_i^k$.
\end{theorem}
\begin{proof}
This is guaranteed by Eq.~\ref{eq:e}, which computes the spin delay based on $\min\{\zeta_i^k, \xi_{i,m}^k\}$. 
Therefore, if $\zeta_i^k \geq \xi_{i,m}^k$ holds for each remote processor $\lambda_m$, all remote requests to $r^k$ are accounted for in $E_i$. Thus, $r^k$ cannot impose any extra blocking on $\tau_i$ given any $P_i^k$. Hence, the theorem holds.
\end{proof}

The benefits of applying the condition at line 4 are: (i) a direct identification of the spin priority that cannot cause any extra blocking on $\tau_i$ and $\text{lhp}(i)$; and (ii) a speed-up of the assignment as some $P^k_i$ are exempted by the initialisation.

\textbf{Search-based Assignment.} 
The assignment starts from the highest-priority task (say $\tau_i$) on $\lambda_m$ (line 11). 
From the perspective of $\tau_i$, a linear search (lines 11 to 22) is conducted to gradually reduce $B_i$ by decreasing $P^k_l$ of all $\tau_l \in \text{llp}(i)$, so that $\tau_i$ can meet its deadline. 
We note that adjusting (\ie, decreasing) $P^k_l$ will not increase the blocking of tasks in $\text{lhp}(i)$, \ie, those that are examined before $\tau_i$.

For each $\tau_i$, function $F^*(\tau_i)$ is firstly constructed to identify the set of global resources required by $\tau_l \in \text{llp}(i)$ and can cause $B_i$ due to a high spin priority, \ie, $P^k_l \geq P_i$ (line 12). The $F^*(\tau_i)$ provides the target set of resources in which the $P^k_l$ can be exploited to reduce $B_i$. 
From lines 13 to 21, a linear search is applied to adjust $P^k_l$ for resources in $F^*(\tau_i)$ with the objective of $R_i \leq D_i$ (where applicable). 

First, we obtain the resource $r^k$ that is currently causing the maximum arrival blocking ($maxk(B_i)$) at line 14. If $r^k \in F^*(\tau_i)$ (line 15), it indicates $B_i$ can be reduced by decreasing $P_l^k$.
In this case, $P^k_l$ is set to $P_i-1$ for all tasks in $\text{llp}(i)$ (line 16), and $r^k$ is removed from $F^*(\tau_i)$ (line 17). 
By doing so, the spinning of $\tau_l$ for $r^k$ cannot delay $\tau_i$ when it arrives.
Otherwise ($r^k \notin F^*(\tau_i)$), the linear search for $\tau_i$ is terminated directly, as $B_i$ is not reducible by adjusting the spin priorities, \eg, $r^k$ is a local resource.
This repeats until $R_i \leq D_i$ or the target resource set $F^*(\tau_i)$ becomes empty. 





With the linear search applied for each task, appropriate spin priorities can be determined by considering the task execution urgency.  
The algorithm finishes after each task in $\Gamma$ is examined, with at most $|\Gamma| \times |\mathbb{R}|$ iterations.

\begin{figure*}[t]
\vspace{-10pt}
\centering
\subfigtopskip=0pt 
\subfigbottomskip=3pt 
\subfigcapskip=-2pt 
\subfigure[Schedulability with varied $N$.]{\label{fig:experiment_fig_task}  
{\includegraphics[width=.6\columnwidth]{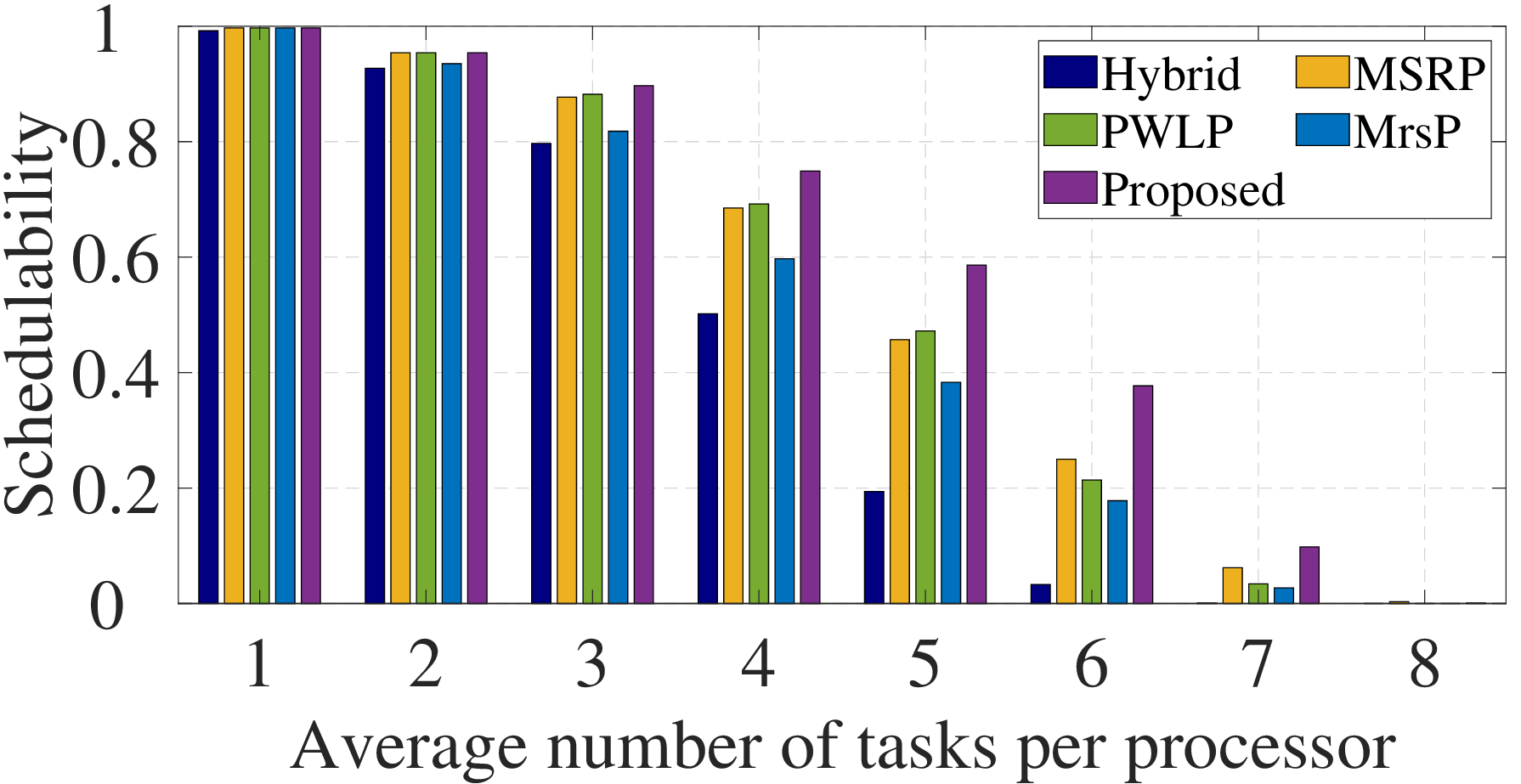}}}
\subfigure[Schedulability with varied $M$.]{\label{fig:experiment_fig_core}  
{\includegraphics[width=.6\columnwidth]{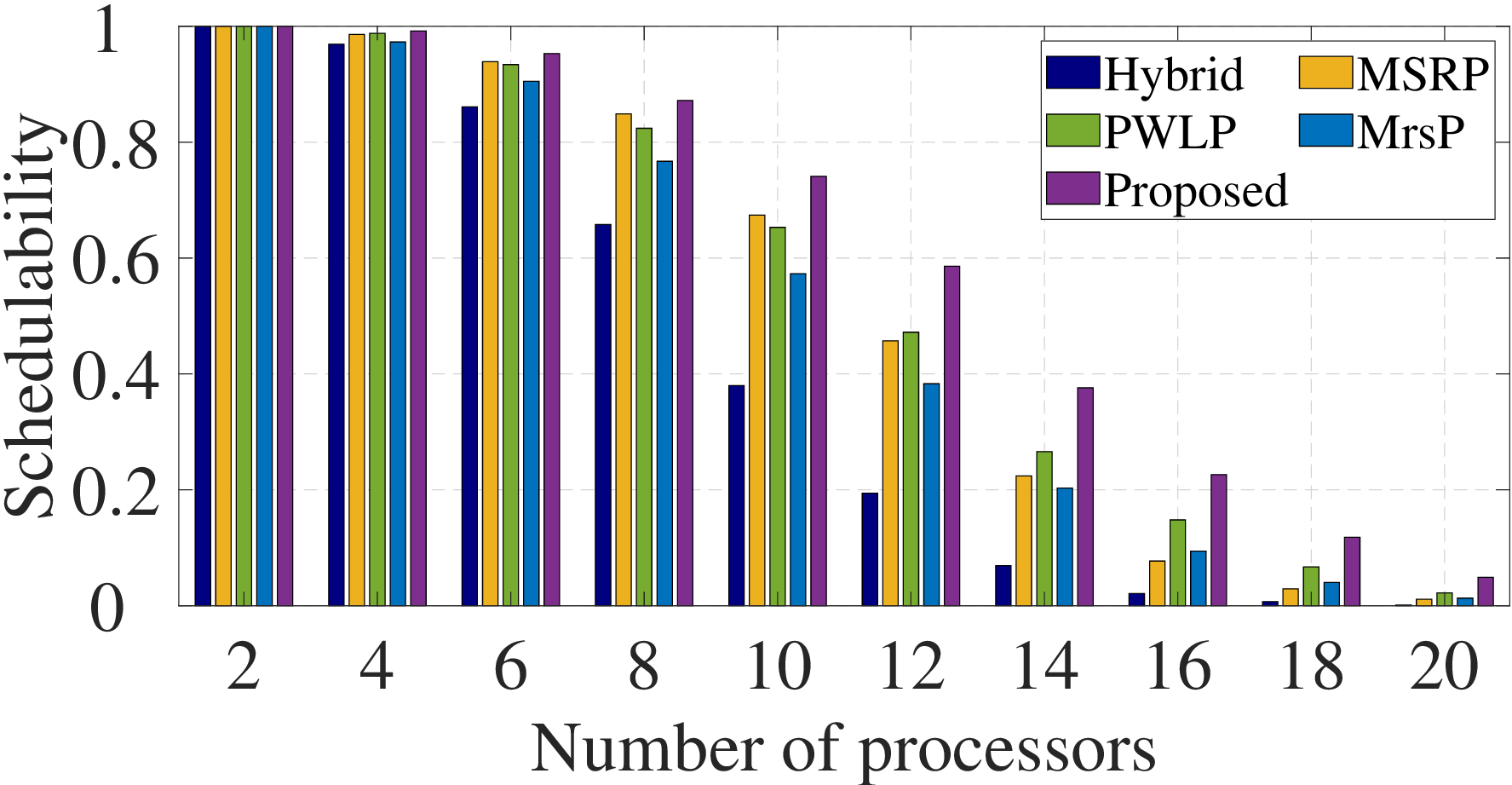}}}
\subfigure[Schedulability with varied $A$.]{\label{fig:experiment_fig_access}  
{\includegraphics[width=.6\columnwidth]{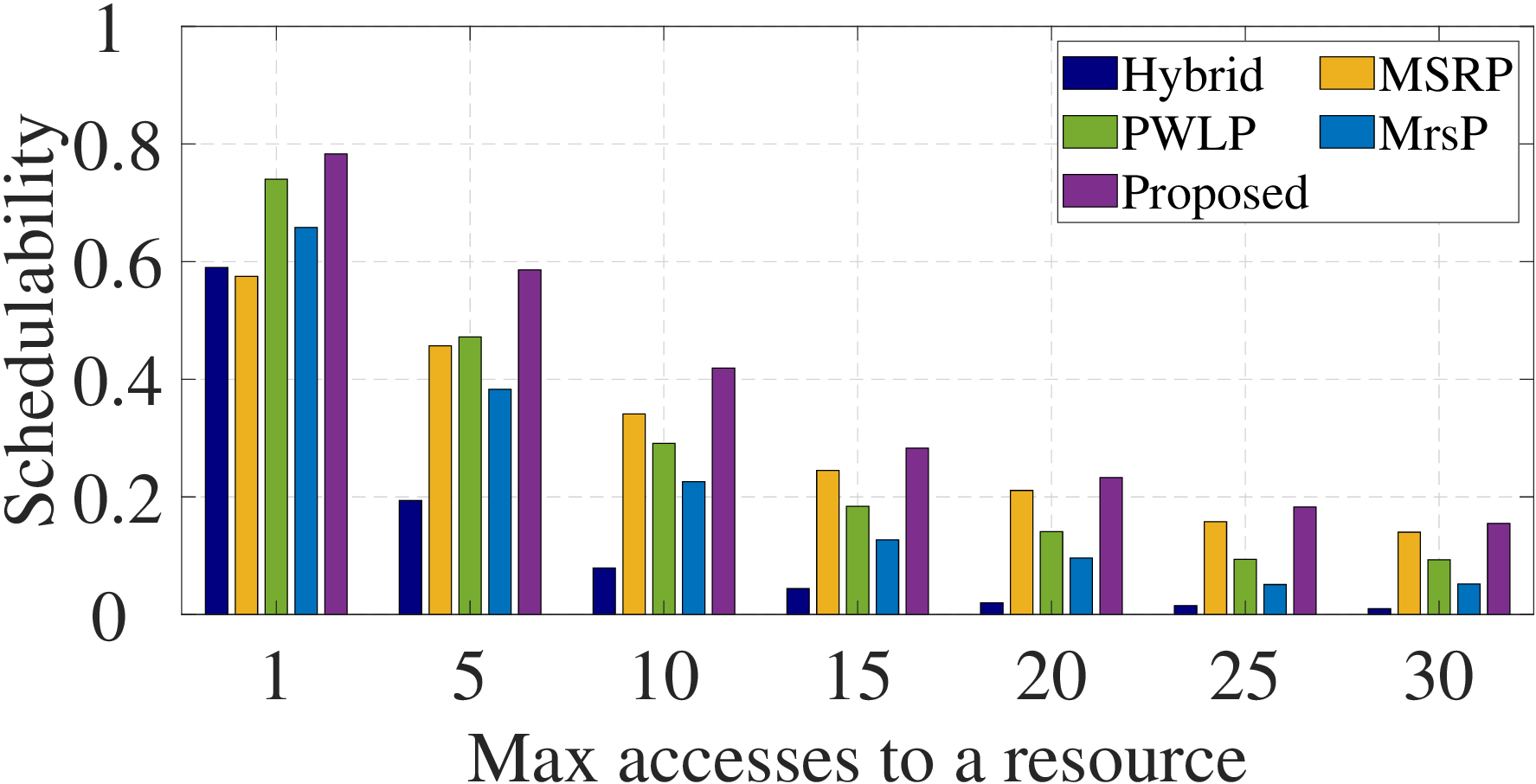}}}
\subfigure[Schedulability with varied $L$.]{\label{fig:experiment_fig_length}  
{\includegraphics[width=.6\columnwidth]{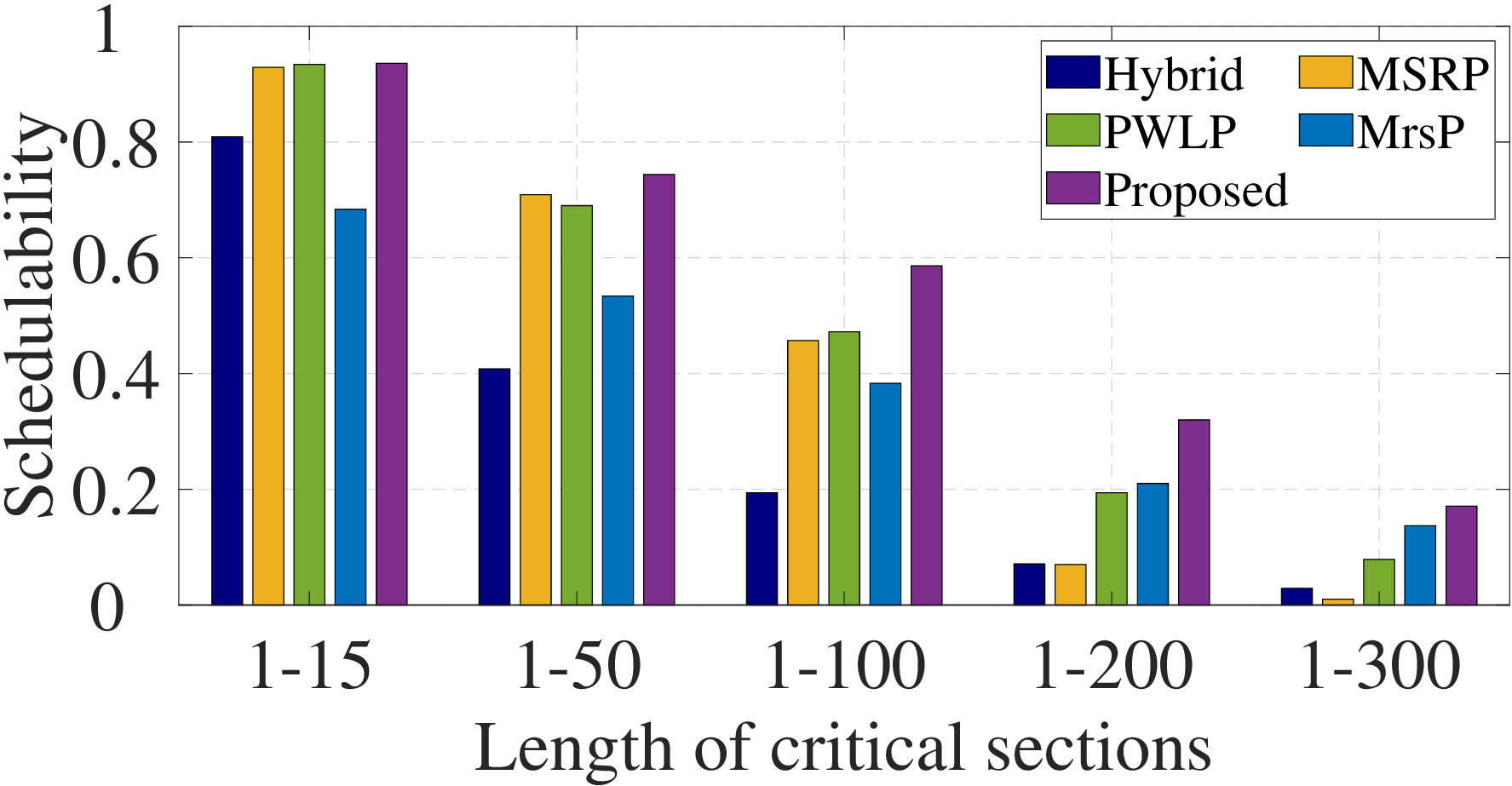}}}
\subfigure[Schedulability with varied $rsf$.]{\label{fig:experiment_fig_rsf}  
\includegraphics[width=.6\columnwidth]{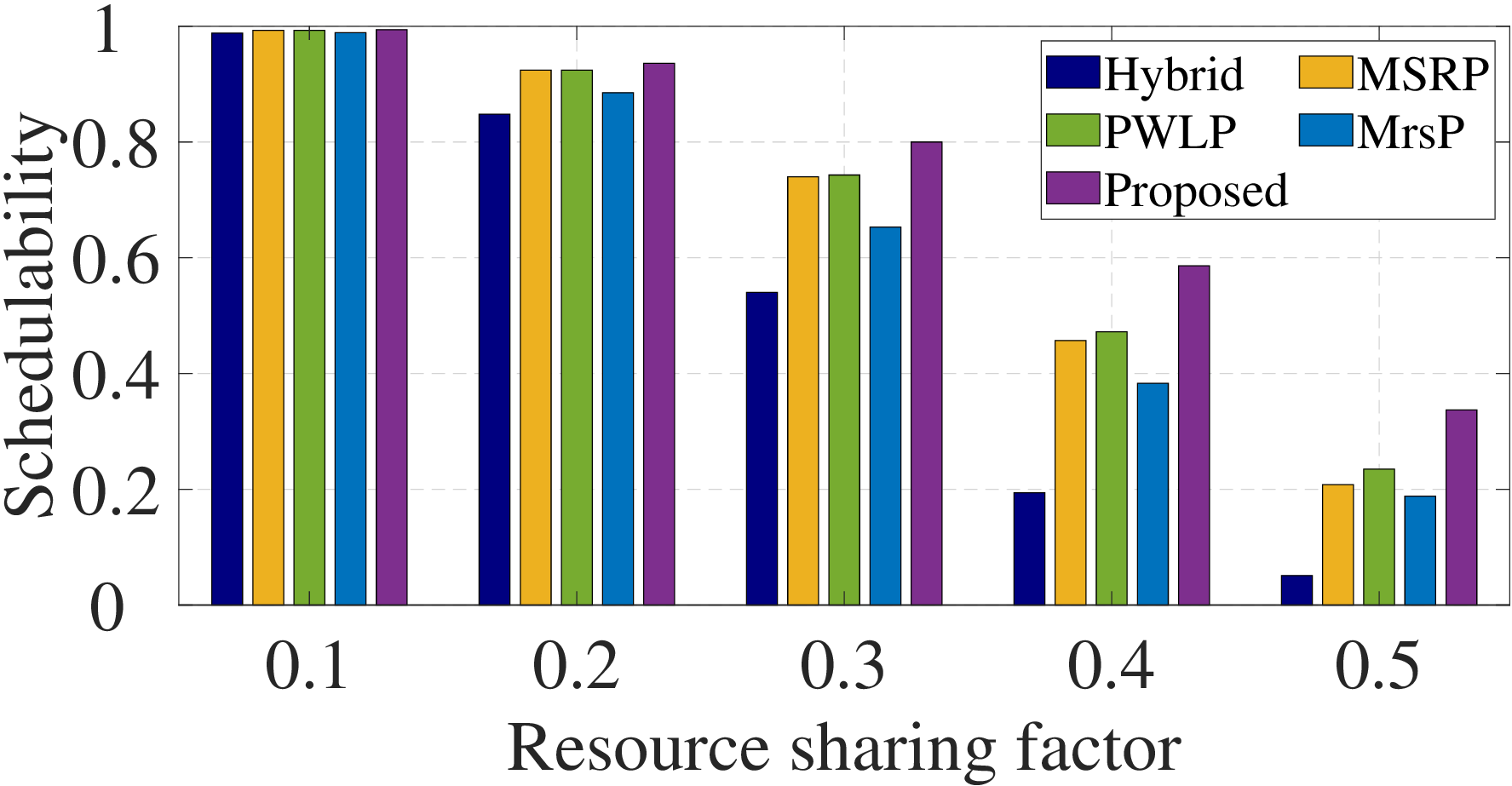}}
\subfigure[Schedulability with varied $K$.]{\label{fig:experiment_fig_resource}  
\includegraphics[width=.6\columnwidth]{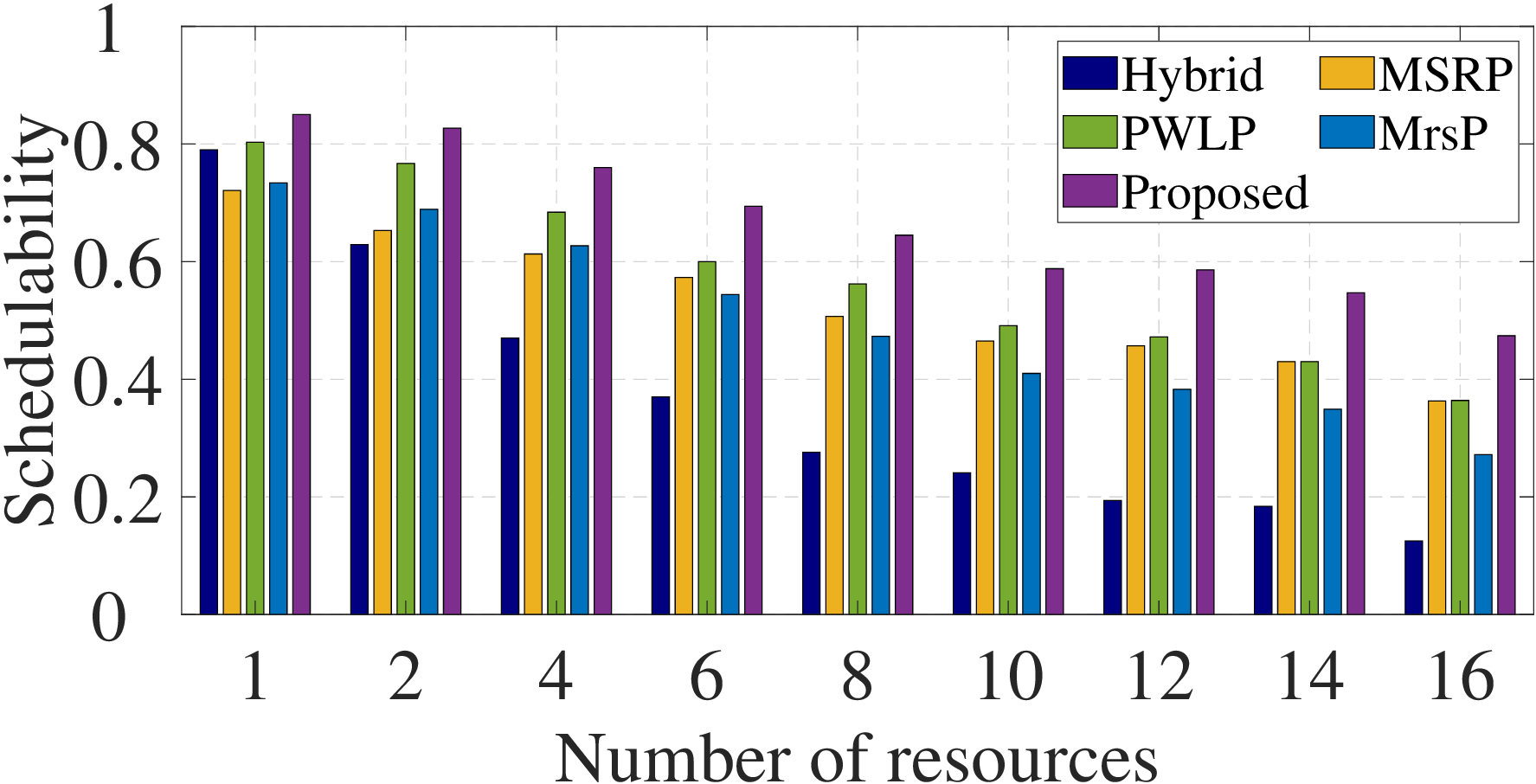}}
\caption{System schedulability with  $N=5$, $M=12$, $A=5$, $L=\left [1 \mu s,100 \mu s \right ] $, $rsf=0.4$, and $K = M$ resources.}
\vspace{-15pt}
\label{fig-4}
\end{figure*}


\subsection{Implementing the Assignment} \label{sec:assignment-approx}
As shown in Alg.~\ref{alg:assign}, extensive computations are required for $\zeta_i^k$, $\xi^k_{i,m}$, $B_i$ (including $maxk(B_i)$), and $R_i$ during the assignment, imposing obstacles for the application of \pname{}. To address this issue, approximations of these values are applied instead of the analysis, avoiding the high computational demand while preserving the effectiveness of the proposed assignment~\cite{zhao2023universal}. We note that various approaches are feasible to approximate the timing bounds, below we present a feasible implementation inspired by the approach applied in~\cite{zhao2023universal}.

First, the condition $\zeta_i^k \geq \xi_{i,m}^k$ at line 4 is approximated as $\phi^k(\tau_i \cup \text{lhp}(i)) \geq \phi^k(\Gamma_m)$, in which 
$\phi^k(\cdot)$ indicates the accessing frequency of $r^k$ by the given tasks, \eg, $\phi^k(\tau_i)=\frac{N_i^k}{T_i}$ for $\tau_i$~\cite{zhao2023universal}. 
Accordingly, 
$\phi^k(\tau_i \cup \text{lhp}(i)) = \sum_{\tau_x \in \tau_i \cup \text{lhp}(i)} \frac{N^k_x}{T_x}  $ and $\phi^k(\Gamma_{m}) = \sum_{\tau_j \in \Gamma_{m}} (\frac{N^k_j}{T_j} +\frac{N^k_j}{T_i}$), respectively.
In general, a higher $\frac{N}{T}$ reflects a higher number of requests. For a remote task $\tau_j$, its back-to-back hit is considered (\ie, $\frac{N^k_j}{T_i}$), which can occur at most once during the release of $\tau_i$. 

In addition, the condition $R_i > D_i$ at line 13 is replaced with $\Psi(\tau_i) > S_i$, where $\Psi(\tau_i)$ approximates the worst-case blocking that $\tau_i$ can incur under the current spin priority configuration and $S_i$ denotes the slack of $\tau_i$~\cite{wieder2013efficient}.
The slack is computed by $S_i = [D_i - \overline{C_i} - \sum_{\tau_h \in \text{lhp}(i)} \left \lceil \frac{T_i}{T_h} \right \rceil \cdot \overline{C_h}]_0$ with a lower bound of 0, indicating the capability of $\tau_i$ for tolerating the blocking while still meeting its deadline. 
The computation of $\Psi(\tau_i)$ is given in Eq.~\ref{eq:psi}, which approximates the blocking of $\tau_i$ based on the accessing frequency of $r^k$ that can impose a spin delay ($\widetilde{e}^k_i$), additional blocking ($\widetilde{w}^k_{i}$), and arrival blocking ($\widetilde{b}^k_i$) to $\tau_i$ within the duration of $T_i$.
\begin{equation} \label{eq:psi}
\small
\Psi(\tau_i)= (
\sum_{r^k \in \mathbb{R}}  \widetilde{e}^k_i \cdot c^k +  
\sum_{r^k \in F^w(\tau_i)} \widetilde{w}^k_{i}\cdot c^k+ 
\underset{r^k \in F^b(\tau_i)}{max} \widetilde{b}^k_i \cdot c^k) \cdot T_i
\end{equation}

Firstly, $\widetilde{e}^k_i$ is computed as $\widetilde{e}^k_i = \sum_{\lambda_m \neq A_i}\min\{\phi^k(\tau_i \cup \text{lhp}(i)), \phi^k(\Gamma_{m})\}$, which follows the same philosophy of Eq.~\ref{eq:e} but with $\phi^k(\cdot)$ applied to avoid extensive computations for response time of tasks.
Using the same approach, $\widetilde{w}^k_i$ and $\widetilde{b}^k_i$ are computed as follows.

The $\widetilde{w}^k_{i}$ is approximated by Eq.~\ref{eq:pw}, based on the frequency of re-requests to $r^k$ issued by $\tau_i \cup \text{lhp}(i)$ (see Sec.~\ref{sec:identifySource}).
As each re-request is caused by a preemption from a $\tau_h$, the frequency of the re-requests caused by $\tau_h$'s preemptions is $\frac{1}{T_h}$.
Thus, given the tasks that can preempt the spinning of $\tau_i \cup \text{lhp}(i)$ (\ie, $\Gamma_i^k=\{\tau_h | P_h > P^k_x, \tau_x \in \tau_i \cup \text{lhp}(i)\}$), the total frequency of the re-requests to $r^k$ that can delay $\tau_i$ is $\sum_{\tau_h \in \Gamma_i^k}\frac{1}{T_h}$.
Accordingly, $\widetilde{w}^k_{i}$ is approximated as the minimum value between $\sum_{\tau_h \in \Gamma_i^k}\frac{1}{T_h}$ and the remaining accessing rate of $r^k$ on each remote processor after the computation of $\widetilde{e}^k_i$.
\begin{equation} \label{eq:pw}
\small
\widetilde{w}^k_{i} = 
\sum_{\lambda_m \neq A_i}
\min\Big\{\sum_{\tau_h \in \Gamma_i^k}\frac{1}{T_h}, [\phi^k(\Gamma_{m}) - \phi^k(\tau_i \cup \text{lhp}(i)) ]_0 \Big\}
\end{equation}

Finally, $\widetilde{b}_i^k$ is computed depending on whether $r^k \in F^b(\tau_i)$ can impose a remote delay (see Eq.~\ref{eq:lb}). 
If $r^k$ only imposes a blocking of one critical section (\ie, the first case in Eq.~\ref{eq:lb}), $\widetilde{b}_i^k = \frac{1}{T_i}$ as only one request can cause the arrival blocking of $\tau_i$. 
However, if $r^k$ can impose a remote delay (\ie, the second case in Eq.~\ref{eq:lb}), $\widetilde{b}_i^k$ is computed by Eq.~\ref{eq:pb} by considering the remaining accessing rate of $r^k$, after computing $\widetilde{e}^k_i$ and $\widetilde{w}^k_{i}$ (\ie, $[\phi^k(\Gamma_m) - \phi^k(\tau_i \cup \text{lhp}(i)) - \sum_{\tau_h \in \Gamma_i^k}\frac{1}{T_h} ]_0$). 
With $\widetilde{b}_i^k$ for all $r^k \in F^b(\tau_i)$ obtained, the $r^k$ that yields the maximum $\widetilde{b}_i^k \cdot c^k$ is returned by the function $maxk(B_i)$ at line 14.
\begin{equation} \label{eq:pb}
\small
\widetilde{b}_i^k =  \frac{1}{T_i} + 
\sum_{\lambda_m \neq A_i} 
\min \Big\{
\frac{1}{T_i}, [\phi^k(\Gamma_{m}) - \phi^k(\tau_i \cup \text{lhp}(i)) - \sum_{\tau_h \in \Gamma_i^k}\frac{1}{T_h} ]_0
\Big\}
\end{equation}

\vspace{-3pt}
\subsection{Summary and Discussion}

This concludes the construction of the spin priority assignment.
We note that the use of approximations can introduce deviations~\cite{zhao2023universal}. However, as shown in Sec.~\ref{sec:evaluation}, this does not undermine the effectiveness of \pname{}. Instead, it provides a cost-effective approach with valuable guidance for the assignment, enhancing the timing performance of \pname{}.

The FRAP can be effectively implemented without introducing significant overhead. In practice, spin priorities produced by Alg.~\ref{alg:assign} can be stored in a lookup table in the user space, \eg, a hash table of ``(task, resource) $\rightarrow$ spin priority". As for the kernel space, the implementation of FRAP is similar to that of PWLP in LITMUS-RT~\cite{calandrino2006litmus, brandenburg2011scheduling, zhao2018thesis}. For a resource request (\eg, an invocation of the \texttt{lock()} function in LITMUS-RT), the additional operations required by FRAP are (i) obtain the corresponding spin priority from the lookup table and pass it as a parameter of \texttt{lock()} and (ii) update the active priority of the task based on R1 to R3 in Sec.~\ref{sec:protocol}. 
However, both operations have a linear complexity, and the overhead does not affect the effectiveness of the protocol.



In addition, the application of the proposed assignment requires detailed knowledge of the system, including $C_i$, $c^k$, and $N_i^k$ for each task and resource.
If such knowledge is not available, \eg, at an early design stage~\cite{zhao2023universal,brandenburg2022multiprocessor}, a simple heuristic can be applied as a pilot design solution, \eg, $P_i^k=\widehat{P}$ for short resources and $P_i^k=P_i$ for long ones~\cite{block2007flexible}. As more system details become available at later phases, the proposed spin priority assignment can be deployed to determine the final system configuration, with the timing verified by the constructed analysis.



\section{Evaluation} 
\label{sec:evaluation}

This section compares the performance of the proposed \pname{} with existing resource sharing solutions. The following approaches and the corresponding schedulability tests are considered as the competing methods: (i) MSRP and its analysis in \cite{wieder2013spin}; (ii) PWLP and its analysis in \cite{zhao2018thesis}; (iii) MrsP and its analysis in \cite{zhao2017new}; and (iv) the hybrid locking approach and its analysis in~\cite{afshar2014flexible,afshar2017optimal} (\ie, \textit{Hybrid}). In addition, an optimisation method is provided in~\cite{zhao2018thesis} that manages each resource by one spin-based protocol (\ie, MSRP, PWLP or MrsP). However, it focuses on the system-level configurations using optimisation techniques, hence, is not considered for comparison. 
For \pname{}, the Primal-Dual algorithm~\cite{ahuja1995network,ford1957primal} is applied as MCMF solver to compute $B_i + W_i$. The spin priorities of tasks are assigned using the approximations described in Sec.~\ref{sec:assignment-approx}.

\subsection{Experimental Setup}
Similar to that in~\cite{wieder2013spin,zhao2017new}, 
the experiments are conducted on systems with $M \in [2, 20]$ processors and $N \cdot M$ tasks with $N \in [1, 8]$, where $N$ gives the average number of tasks per processor. The periods of tasks are randomly chosen from a log-uniform distribution over $[1ms, 1000ms]$ with $D_i = T_i$. The utilisation $U_i$ of each task is generated by the UUniFast-Discard algorithm~\cite{emberson2010techniques} with a bound of $0.1 \cdot M \cdot N$. The total execution time of $\tau_i$ is $\overline{C_i} = U_i \cdot T_i$. The task priority is assigned by the deadline-monotonic priority ordering and the allocation is produced by the worst-fit heuristic. 
The system contains $K$ shared resources and each resource has a length of critical section randomly decided in a given range $L=\left [ 1\mu s,300\mu s \right ] $, covering a wide range of realistic applications~\cite{wieder2013spin,zhao2017new}. A resource sharing factor $rsf \in \left[0.1,0.5 \right]$ is used to control the number of tasks that can access resources, in which a task can require a random number (up to $K$) of shared resources. The number of accesses from a task to a resource is randomly generated from the range $\left[1, A\right]$ with $A = \left[1,30\right]$. Let $C_i^r$ denote the total resource execution time of $\tau_i$, we enforce that $C_i=\overline{C_i}-C_i^r \ge 0$ for each generated task. For the analysis of MrsP, the cost of one migration is set to $6 \mu s$, as measured by reported by~\cite{zhao2017new} under the Linux operating system.

\subsection{Overall System Schedulability}
Fig.~\ref{fig-4} presents the overall system schedulability of the evaluated protocols, in which 1000 systems are randomly generated for each system configuration.

\textbf{Observation 1:} \pname{} constantly outperforms all competing methods in terms of system schedulability.

This observation is obtained by Fig.~\ref{fig:experiment_fig_task} to~\ref{fig:experiment_fig_resource}. For the \textit{Hybrid}, it is outperformed by all other protocols due to the limitations described in Sec.~\ref{sec:related-protocol}. In particular, the overly pessimistic analytical bounds significantly undermine the effectiveness of this approach. In contrast, the proposed \pname{} addresses these limitations and achieves an improvement of $76.47\%$ on average (up to $177.52\%$) in system schedulability.
Compared to MSRP, PWLP, and MrsP, \pname{} shows a constantly higher schedulability across all scenarios with an improvement of $17.87\%$, $15.20\%$ and $32.73\%$ on average (up to $65.85\%$ compared with MrsP), respectively. In particular, the \pname{} demonstrates a strong schedulability when the scheduling pressure is relatively high, \eg, $N\geq5$ in Fig.~\ref{fig:experiment_fig_task}, $M\geq10$ in Fig.~\ref{fig:experiment_fig_core}, and $rsf \geq 0.4$ in Fig.~\ref{fig:experiment_fig_rsf}.
This justifies the effectiveness of \pname{}, which provides flexible and fine-grained resource control with improved schedulability.

In addition, as expected, 
the performance of existing protocols is highly sensitive to certain system characteristics due to the fixed resource accessing rules.
MSRP shows a strong performance with intensive resource requests (\eg, $A\geq10$ in Fig.~\ref{fig:experiment_fig_access}) or short critical sections (\eg, $L=[1\mu s,15\mu s]$ in Fig.~\ref{fig:experiment_fig_length}). 
However, its performance is significantly compromised when facing long resources, \eg, $L=[1\mu s,200\mu s]$ in Fig.~\ref{fig:experiment_fig_length}.
Besides, PWLP becomes suitable for systems with a relatively low $A$ (Fig.~\ref{fig:experiment_fig_access}) or $K$ (Fig.~\ref{fig:experiment_fig_resource}),
whereas MrsP is favourable if very long resources exist, \eg, $L=[1\mu s,300\mu s]$ in Fig.~\ref{fig:experiment_fig_length}.
By contrast, \pname{} overcomes such limitations by assigning appropriate spin priorities with task execution urgency taken into account, providing the highest schedulability under various resource accessing scenarios. 

\begin{table}[t]
\vspace{-10pt}
\caption{Percentage of schedulable system with varied $A$.}
\label{tab:experiment-access}
\vspace{-2pt}
\setlength{\tabcolsep}{1pt}
\resizebox{1\columnwidth}{!}{
    \begin{tabular}{c|cc|cc|cc|cc}
        \multirow{2}{*}{$A=$} & {\pname{}}$\&$        &!{\pname{}}      & {\pname{}}$\&$         & !{\pname{}}      & {\pname{}}$\&$         & !{\pname{}}         & {\pname{}}$\&$         & !{\pname{}}  \\
        &  {!MSRP}       &$\&${MSRP}       & {!PWLP}     &$\&${PWLP}  &  {!MrsP}     &$\&${MrsP} &  !{Hybrid}     &$\&${Hybrid} \\

        \hline
        1 & \textbf{20.67} & 0.00 & \textbf{4.38} & 0.17 & \textbf{14.43} & 0.04 & \textbf{20.76} & 0.45 \\
        5 & \textbf{12.42} & 0.06 & \textbf{11.68} & 0.02 & \textbf{19.93} & 0.10 & \textbf{37.99} & 0.17 \\
        10 & \textbf{7.32} & 0.12 & \textbf{12.72} & 0.05 & \textbf{20.02} & 0.04 & \textbf{33.82} & 0.07\\
        20 & \textbf{3.11} & 0.16 & \textbf{10.09} & 0.01 & \textbf{15.27} & 0.04 & \textbf{21.51} & 0.05\\
        30 & \textbf{1.62} & 0.16 & \textbf{7.87} & 0.01 & \textbf{12.09} & 0.01 & \textbf{15.30} & 0.01\\
    \end{tabular}
}
\end{table}

\begin{table}[t]
\caption{Percentage of schedulable system with varied $L$.}
\label{tab:experiment-csl}
\vspace{-2pt}
\setlength{\tabcolsep}{1pt}
\resizebox{1\columnwidth}{!}{
\begin{tabular}{c|cc|cc|cc|cc}
\multirow{2}{*}{$L=$} & {\pname{}}$\&$        &!{\pname{}}      & {\pname{}}$\&$         & !{\pname{}}      & {\pname{}}$\&$         & !{\pname{}}         & {\pname{}}$\&$         & !{\pname{}}  \\
        &  !{MSRP}       &$\&${MSRP}       & !{PWLP}     &$\&${PWLP}  &  !{MrsP}     &$\&${MrsP} &  !{Hybrid}     &$\&${Hybrid} \\

\hline
1-15              & \textbf{0.88}          & 0.01           & \textbf{0.22}              & 0.02             & \textbf{25.75}             & 0.00    & \textbf{12.99}  & 0.06\\
1-50              & \textbf{3.27}          & 0.00           & \textbf{4.86}              & 0.05             & \textbf{22.40}             & 0.00    & \textbf{34.32}  & 0.12\\
1-100             & \textbf{12.42}         & 0.06           & \textbf{11.68}              & 0.02             & \textbf{19.93}             & 0.10    & \textbf{37.99} & 0.17\\
1-200             & \textbf{25.42}         & 0.02           & \textbf{13.18}              & 0.03             & \textbf{11.64}             & 0.91    & \textbf{24.71}  & 0.22\\
1-300             & \textbf{15.67}        & 0.03           & \textbf{9.18}              & 0.08             & \textbf{6.36}              & 2.03  & \textbf{14.11} & 0.22\\
\end{tabular}}
\vspace{-10pt}
\end{table}

\subsection{The Effectiveness of Proposed Protocol}
The above shows the overall schedulability of \pname{}. This section further investigates the performance of \pname{} and existing protocols by examining the percentage of systems that are feasible under protocol X but are not schedulable with protocol Y in 10,000 systems, represented as X$\&$!Y.

\textbf{Observation 2:}
\pname{} can schedule $10.28\%$, $8.59\%$, $16.78\%$ and $25.35\%$ systems on average that are infeasible under MSRP, PWLP, MrsP and \textit{Hybrid}, respectively.

This is observed in Tab.~\ref{tab:experiment-access} and Tab.~\ref{tab:experiment-csl}. As shown by the results, \pname{} in general shows a dominating performance compared to existing protocols, where it can schedule a large percentage of systems that are infeasible with existing protocols. 
For instance, with $A=5$ in Tab.~\ref{tab:experiment-access}, over $11.68\%$ of systems are feasible with \pname{} but are unschedulable using an existing protocol, whereas at most $0.17\%$ of systems are schedulable under existing protocols while \pname{} cannot.
However, MrsP shows a relatively strong performance with $L=\left [ 1\mu s,300\mu s \right ] $ in Tab.~\ref{tab:experiment-csl}, which can schedule $2.03\%$ systems that are infeasible under \pname{}. 
The reason is with MrsP applied, the independent high-priority tasks are free from the arrival blocking due to the ceiling priority mechanism (see Sec.~\ref{sec:related-protocol}). 
However, such tasks under \pname{} can incur a local blocking as resources are executed non-preemptively. 
In addition, although the approximations in Sec.~\ref{sec:assignment-approx} can introduce deviations,
such an impact is trivial as suggested by the results. This validates the effectiveness of the proposed implementation of the spin priority assignment.

\subsection{Computation Cost of Proposed Analysis}
This section compares the computation cost of the proposed analysis, the holistic analysis in~\cite{zhao2018thesis} and the ILP-based analysis in~\cite{wieder2013spin} under MSRP and PWLP.
The holistic and ILP-based methods directly support the analysis of MSRP and PWLP~\cite{wieder2013spin,zhao2018thesis}.
In addition, the proposed MCMF-based analysis is capable of handling both MSRP and PWLP with the corresponding spin priority assigned, \ie, $\widehat{P}$ for MSRP and $P_i$ for PWLP.
The measurements are collected on an Intel i5-13400 processor with a frequency of 2.50GHz.

\textbf{Observation 3:} The cost of the proposed analysis is reduced by 12.18x on average compared to ILP-based analysis.

This is observed in Tab.~\ref{tab:cost}. 
With an increase in $N$ and $M$, the cost of all analysing methods grows due to the increasing system complexity. However, a decrease is observed for all methods when the scheduling pressure is high (\eg, $M=16$), where infeasible tasks are more likely to be found that directly terminate the analysis. 
As expected, the holistic analysis shows the lowest cost without any optimisation techniques, whereas the ILP-based one has the highest computation time.
As for the constructed analysis, it shows a much lower cost compared to the ILP-based analysis in all cases. This justifies the use of MCMF in the proposed analysis, which avoids the high computational demand and mitigates the scalability issue in ILP, while producing tight analytical results.



\begin{table}[t]
\vspace{-10pt}
    \caption{Cost of analysis (in $ms$) with varied $N$ and $M$.}
    \label{tab:cost}
    \vspace{-2pt}
    \centering
    \setlength{\tabcolsep}{8pt}
    \resizebox{\columnwidth}{!}{
        \begin{tabular}{c|cc|cc|cc}

        \multirow{2}{*}{$N=$} & \multicolumn{2}{c|}{{Holistic~\cite{zhao2018thesis}}} & \multicolumn{2}{c|}{ILP-based~\cite{wieder2013spin}}  & \multicolumn{2}{c}{Proposed}   \\
            \cline{2-7}
            &\raisebox{-0.3ex}[0pt][0pt]{\textit{MSRP}}  & \raisebox{-0.3ex}[0pt][0pt]{\textit{PWLP}}  & \raisebox{-0.3ex}[0pt][0pt]{\textit{MSRP}}  & \raisebox{-0.3ex}[0pt][0pt]{\textit{PWLP}}  & \raisebox{-0.3ex}[0pt][0pt]{\textit{MSRP}}  & \raisebox{-0.3ex}[0pt][0pt]{\textit{PWLP}} \\
            \hline
           
            1 & 0.05 & 0.03 & 3.36 & 3.46 & \textbf{0.15} & \textbf{0.10} \\
            3 & 0.56 & 0.33 & 75.62 & 70.76 & \textbf{2.10} & \textbf{6.54}  \\
            5 & 1.33 & 0.80 & 241.84 & 218.71 & \textbf{5.01} & \textbf{43.17}  \\
            7 & 0.73 & 0.39 & 259.97 & 252.53 & \textbf{2.68} & \textbf{14.61}  \\
            \hline
        \multirow{2}{*}{$M=$} & \multicolumn{2}{c|}{{Holistic~\cite{zhao2018thesis}}} & \multicolumn{2}{c|}{ILP-based~\cite{wieder2013spin}}  & \multicolumn{2}{c}{Proposed}   \\
            \cline{2-7}
            &\raisebox{-0.3ex}[0pt][0pt]{\textit{MSRP}}  & \raisebox{-0.3ex}[0pt][0pt]{\textit{PWLP}}  & \raisebox{-0.3ex}[0pt][0pt]{\textit{MSRP}}  & \raisebox{-0.3ex}[0pt][0pt]{\textit{PWLP}}  & \raisebox{-0.3ex}[0pt][0pt]{\textit{MSRP}}  & \raisebox{-0.3ex}[0pt][0pt]{\textit{PWLP}} \\
            \hline
 
            4 & 0.11 & 0.10 & 19.37 & 17.52 & \textbf{0.71} & \textbf{7.95}  \\
            8 & 0.76 & 0.44 & 128.30 & 113.22 & \textbf{3.24} & \textbf{32.82}  \\
            12 & 1.33 & 0.80 & 242.41 & 219.49 & \textbf{5.14} & \textbf{44.17}  \\ 
            16 & 0.61 & 0.55 & 229.79 & 225.40 & \textbf{2.05} & \textbf{20.21}  \\
           
        \end{tabular}
    }
    \vspace{-10pt}
\end{table}

\subsection{Summary}
In summary, \pname{} outperforms existing resource sharing solutions with a dominating performance in terms of schedulability, and requires much less computation cost for the timing analysis.
With \pname{} constructed, we provide a flexible and fine-grained resource sharing solution for FP-FPS systems, in which the effectiveness holds for a wide range of resource accessing scenarios, addressing the limitations of existing protocols with improved performance. 

\section{Conclusion} 
\label{sec:conclusion}

In the Buddhism philosophy, one universal suffering is ``not getting what one needs'', which is valid for computing systems.
This paper presents \pname{}, a flexible resource accessing protocol for FP-FPS systems with FIFO spin locks. Instead of rigid resource control, \pname{} enables flexible spinning in which a task can spin at any priority level in a range for accessing a resource.
A novel MCMF-based analysis is constructed that produces tight blocking bounds for \pname{}.
In addition, a spin priority assignment is deployed that exploits flexible spinning to enhance the performance of \pname{}.
Experimental results show that \pname{} outperforms existing resource sharing solutions with spin locks in schedulability, with less computation cost needed by its analysis.
Future work will extend \pname{} to support more complex scenarios, \eg, nested resource accesses with non-uniform worst-case computation time~\cite{ward2013fine}. 


\balance
\bibliographystyle{IEEEtran}
\bibliography{ref}

\end{document}